\documentclass[11pt]{article}
\usepackage[letterpaper, left=1in, right=1in, bottom=1in, top=1in]{geometry}
\usepackage[T1]{fontenc}
\usepackage[dvipdfmx]{graphicx}
\graphicspath{{../Figure/}}
\usepackage{hyperref}
\usepackage{color}

\usepackage{authblk}
\usepackage{amsthm,amsmath,amssymb}
\usepackage{thm-restate}
\usepackage{mathtools}
\usepackage{algpseudocode}
\usepackage{algorithm}

\newcommand{\lcp}{\mathit{lcp}}
\newcommand{\Bord}{\mathsf{Bord}}
\newcommand{\rle}{\mathrm{rle}}

\newcommand{\lub}{\ell^\star}
\newcommand{\LUB}{\mathit{LUB}}
\newcommand{\RmBordf}{\textproc{rm-long-bordered}}
\newcommand{\Candf}{\textproc{candidate}}
\newcommand{\LongSUBf}{\textproc{longest-short-ub}}
\newcommand{\BG}{\mathsf{BG}}
\newcommand{\pp}{\mathit{pp}}
\newcommand{\first}{\mathsf{first}}
\newcommand{\last}{\mathsf{last}}
\newcommand{\tRLESA}{\mathsf{tRLESA}}
\newcommand{\tRLELCP}{\mathsf{tRLELCP}}
\newcommand{\tRLEISA}{\mathsf{tRLESA}^{-1}}
\newcommand{\Beg}[1]{\mathsf{beg}_{#1}}
\newcommand{\End}[1]{\mathsf{end}_{#1}}
\newcommand{\Exp}[1]{\mathsf{exp}_{#1}}
\newcommand{\RLElcp}[4]{\mathsf{rlcp}(#1,#2,#3,#4)}
\newcommand{\LongestPref}[4]{\mathsf{ridx}_{#1,#2}(#3,#4)}
\newcommand{\chara}[1]{\mathsf{char}(#1)}
\newcommand{\RLEB}{\mathsf{rBord}}
\newcommand{\RLEBsub}{\mathsf{rBord^\star}}
\newcommand{\CheckLBf}{\textproc{check-long-border}}
\newcommand{\CompLUBf}{\textproc{compute-lub}}
\newcommand{\LongUBf}{\textproc{longest-ub}}
\newcommand{\ExpSum}{\mathsf{ExpSum}}
\newcommand{\tRLELCPr}{\mathsf{tRLELCPr}}
\newcommand{\interval}{\mathrm{Int}}

\newtheorem{lemma}{Lemma}

\newtheorem{example}{Example}
\newtheorem{corollary}{Corollary}
\newtheorem{observation}{Observation}
\newtheorem{claim}{Claim}
\newtheorem{definition}{Definition}

\begin{document}

\title{Algorithms for Computing Longest Unbordered Factors on Run-Length Encoded Strings}
\author[1]{Shoma Sekizaki}
\author[1]{Takuya Mieno}
\affil[1]{The University of Electro-Communications, Japan}
\affil[ ]{\texttt{tmieno@uec.ac.jp}}
\date{}
\maketitle
\begin{abstract}
  A border of a string is a non-empty proper prefix of the string
  that is also a suffix.
  A string is called unbordered if it has no border.
  The longest unbordered factor is a fundamental notion in stringology,
  closely related to string periodicity.
  This paper addresses the longest unbordered factor problem:
  given a string of length $n$, the goal is to compute
  its longest factor that is unbordered.
  While recent work has achieved subquadratic and near-linear time algorithms for this problem,
  the best known worst-case time complexity remains $O(n \log n)$~[Kociumaka et al., ISAAC 2018].
  In this paper, we investigate the problem in the context of compressed string processing, focusing specifically on run-length encoded (RLE) strings.
  We first present a simple yet crucial structural observation relating unbordered factors and RLE-compressed strings.
  Building on this, we propose an algorithm that solves the problem in $O(m^{1.5} \log^2 m)$ time and $O(m \log^2 m)$ space,
  where $m$ is the size of the RLE-compressed input string.
  Furthermore, we propose another algorithm for the problem
  that runs in average $O(m \log^2 m)$ time and $O(m\log m)$ space:
  The average-case bound holds under an exponent-bounded random-input model,
  in which the input is chosen uniformly from the strings of RLE size $m$ whose run exponents are at most a given bound $E\geq 2$.
  Our approaches extend key ideas from [Gawrychowski et al., SPIRE 2015] to the RLE setting through new combinatorial insights. 
  When the RLE size $m$ is sufficiently small compared to $n$, our algorithms achieve linear-time behavior in $n$,
  potentially leading to improved performance over existing methods in such cases.
\end{abstract}

\section{Introduction}\label{sec:intro}

A non-empty string $b$ is called a \emph{border} of another string $T$ if $b$ is both a prefix and a suffix of $T$.
A string is said to be \emph{bordered} if it has a border, and \emph{unbordered} otherwise.
Unbordered factors are known to have a deep connection with the smallest period of the string.
The length of a string $uv$ is called a \emph{period} of a string $T$ if $T = (uv)^k u$ for some strings $u, v$ and an integer $k \geq 1$.
The concept of string periodicity is fundamental and has applications in various areas of string processing,
including pattern matching, text compression, and sequence assembly in bioinformatics~\cite{KnuthMP77,CrochemoreMRS99,MargaritisS95}.

In 1979, Ehrenfeucht and Silberger~\cite{EhrenfeuchtS79} posed the problem of determining the conditions under which $\tau(T) = \pi(T)$ holds
for a string $T$ of length $n$, where $\tau(T)$ denotes the length of the longest unbordered factor of $T$ and $\pi(T)$ denotes the smallest period of $T$.
They further conjectured that $\tau(T) \leq n/2$ implies $\tau(T) = \pi(T)$.
However, this conjecture was disproved by Assous and Pouzet~\cite{AssousP79},
who provided a counterexample.
Subsequently, some progress was made toward the conjecture~\cite{Duval82,MignosiZ02,DuvalHN08,HarjuN04,Holub05,HarjuN07}.
Finally, in 2012, Holub and Nowotka~\cite{HolubN12} solved this longstanding open problem,
showing that $\tau(T) = \pi(T)$ holds if $\tau(T) \leq 3n/7$, and that this bound is tight
due to the counterexample of~\cite{AssousP79}.

This result led to increased research activity
on algorithms for computing the longest unbordered factor~\cite{DuvalLL14,LoptevKS15,GawrychowskiKSS15,KociumakaKMP18}.
As a special case, when a string of length $n$ is periodic (i.e., its smallest period is at most $n/2$),
its longest unbordered factor can be computed in $O(n)$ time~\cite{DuvalLL14}.
Unfortunately, since many strings are non-periodic on average~\cite{LoptevKS15}, this linear-time approach has limited applicability.
For the general case, a straightforward $O(n^2)$-time algorithm can be designed by constructing \emph{border arrays}~\cite{KnuthMP77},
which can be computed in linear time, for all suffixes of a string $T$ of length $n$.
The resulting $n$ border arrays indicate whether each factor of $T$ is bordered or unbordered.
The first non-trivial algorithm for computing the longest unbordered factor was proposed by Loptev et al.~\cite{LoptevKS15},
who presented an algorithm with average-case running time $O(n^2 / \sigma^4)$,
where $\sigma$ is the alphabet size.
Furthermore, Cording et al.~\cite{CordingGKK21} proved that the expected length of the longest unbordered factor of a random string is $n - O(\sigma^{-1})$,
and used this result to propose an average-case $O(n)$-time algorithm.
In terms of worst-case time complexity, all of the above are quadratic-time algorithms.
The first worst-case subquadratic-time algorithm was given by Gawrychowski et al.~\cite{GawrychowskiKSS15},
whose algorithm runs in $O(n\sqrt{n})$ time.
We will later review the basic strategy of their algorithm and explain how to extend it to the RLE setting.
The state-of-the-art algorithm for this problem is an $O(n \log n)$-time algorithm proposed by Kociumaka et al.~\cite{KociumakaKMP18,KociumakaRRW24},
which exploits combinatorial properties of unbordered factors and sophisticated data structures,
including the \emph{prefix-suffix query} (PSQ) data structure.
As of 2018~\cite{KociumakaKMP18}, their algorithm was reported to run in $O(n \log n \log^2 \log n)$ time in the worst case
due to the cost of constructing the PSQ data structures.
Later improvements~\cite{KociumakaRRW24} sped up the construction of the PSQ data structure to linear time,
bringing the overall algorithm down to $O(n \log n)$ time.
Whether the longest unbordered factor can be computed in $O(n)$ time remains open. 

In this paper, instead of directly aiming to an $O(n)$-time algorithm,
we propose efficient solutions in the context of \emph{compressed string processing}.
Especially, this work focuses on the \emph{run-length encoding} (RLE) of a string.
We first give a simple but important relationship between unbordered strings and RLE strings (Lemma~\ref{lem:unbordered}).
Using this relationship, we propose an RLE-based algorithm for computing all longest unbordered factors that runs in $O(m \sqrt{m} \log^2 m)$ time and uses $O(m \log^2 m)$ space,
where $m$ is the size of the RLE-compressed string.
When $m$ is sufficiently small (e.g., $m < n^{2/3 - \varepsilon}$ for a small constant $\varepsilon > 0$), our approach achieves $O(n)$ time,
thus improving the worst-case complexity over existing methods for such cases.
On the one hand, the high-level idea of our approach is inspired by the algorithm of Gawrychowski et al.~\cite{GawrychowskiKSS15},
which achieves subquadratic time via a non-trivial combination of fundamental string data structures and combinatorial techniques on strings.
On the other hand, our algorithm differs in details and requires techniques specially tailored to unbordered factors in RLE strings,
particularly in Sections~\ref{sec:candidate} and~\ref{sec:longbordercheck}.
Furthermore, in Section~\ref{sec:mlog2m}, we propose another RLE-based algorithm that runs
in \emph{average} $O(m\log^2 m)$ time, and uses $O(m \log m)$ space,
under an exponent-bounded random-input model,
in which the input is chosen uniformly from the strings of RLE size $m$
whose run exponents are at most a given bound $E \ge 2$.

A preliminary version of this article was presented at SPIRE 2025~\cite{SekizakiM25}. 
This current article includes full proofs of lemmas omitted in the preliminary version,
a linear-time algorithm for computing a sparse border array w.r.t. the RLE~(Lemma~\ref{lem:lbordlinear}),
and an average $O(m\log^2 m)$-time algorithm~(Section~\ref{sec:mlog2m}) as new materials.

\section{Preliminaries}\label{sec:prep}
Let $\Sigma$ be an alphabet.
An element in $\Sigma$ is called a character.
An element in $\Sigma^\star$ is called a string.
The empty string is the string of length $0$, which is denoted by $\varepsilon$.
A string in which all characters are identical is called a unary string.
The concatenation of strings $S$ and $T$ is written as $S\cdot T$, or simply $ST$ when there is no confusion.
Let $T$ be a non-empty string.
If $T = X\cdot Y\cdot Z$ holds for some strings $X, Y$, and $Z$,
then $X$, $Y$, and $Z$ are called a prefix, a factor, and a suffix of $T$, respectively.
Further, they are called a proper prefix, a proper factor, and a proper suffix of $T$
if $X \ne T$, $Y \ne T$, and $Z \ne T$, respectively.
A non-empty string $B$ is called a border of $T$
if $B$ is both a proper prefix and a proper suffix of $T$.
We call the occurrence of $B$ as a prefix (resp.~suffix) of $T$
the prefix-occurrence (resp.~suffix-occurrence) of border $B$.
We call the longest border of $T$ \emph{the} border of $T$.
A string $T$ is said to be bordered if $T$ has a border,
and is said to be unbordered otherwise.
We denote by $|T|$ the length of $T$.
For an integer $i$ with  $1\le i\le |T|$,
we denote by $T[i]$ the $i$th character of $T$.
For integers $i, j$ with  $1\le i \le j \le |T|$,
we denote by $T[i.. j]$ the factor of $T$ that starts at position $i$ and ends at position $j$.
For strings $S$ and $T$, we denote by $\lcp(S, T)$ the length of the longest common prefix (in short, lcp) of $S$ and $T$.
An integer $p$ with $1\le p \le |T|$ is called a period of $T$
if $T[i] = T[i+p]$ holds for all $i$ with $1 \le i \le |T|-p$.
We call the smallest period of $T$ \emph{the} period of $T$.
The \emph{border array} $\Bord_T$ of a string $T$ is an array of length $|T|$,
where $\Bord_T[i]$ stores the length of the border of $T[1..i]$ for each $1 \leq i \leq |T|$~\cite{KnuthMP77}.
The \emph{border-group array} $\BG_T$ of a string $T$ is an array of length $|T|$ such that, for each $1 \le i \le |T|$,
$\BG_T[i]$ stores the length of the shortest border of $T[1.. i]$ whose smallest period equals that of $T[1..i]$ if such a border exists,
and $\BG_T[i] = i$ otherwise~\cite{MitaniMSH24}.
The term \emph{border-group} reflects the following fact:
\begin{lemma}[\cite{KnuthMP77}]\label{lem:periodofborder}
  The set of borders of a string $w$ can be partitioned into $O(\log |w|)$ groups according to their smallest periods.
  The lengths of borders in each group form a single arithmetic progression and the common difference is the smallest period.
\end{lemma}
Given any border length of (a prefix of) $T$, the border-group array allows us to access the smallest border length in the group containing it in constant time.
A concrete example of the arrays is shown in Figure~1 of~\cite{MitaniMSH24}.
It is known that the border array and the border-group array of a string $T$
can be computed by $O(|T|)$ character comparisons~\cite{KnuthMP77,MitaniMSH24}.

The \emph{range maximum query} (\emph{RMQ}) over an integer array $A$ of length $N$ is,
given a query range $[i, j] \subseteq [1, N]$,
to output a position $p$ such that $A[p]$ is a maximum value among sub-array $A[i.. j]$.
The \emph{range minimum query} (\emph{RmQ}) is defined analogously.
The following result is known.
\begin{lemma}[e.g.,~\cite{BenderF00}]\label{lem:rmq}
  There is a data structure of size $O(N)$ that can answer
  any RMQ (and RmQ) over an integer array of length $N$ in $O(1)$ time.
  The data structure can be constructed in $O(N)$ time.
\end{lemma}

In what follows, we fix an arbitrarily \emph{non-unary} string $T$ of length $n$ for our purpose.
This is because the longest unbordered factor of a unary string $\mathtt{a}\cdots\mathtt{a}$ is simply $\mathtt{a}$.

\section{Tools for RLE strings}
This section provides some tools for RLE strings
that are commonly used in subsequent sections.

\subsection{Run-Length Encoding; RLE}
The Run-Length Encoding (RLE) of a string $T$,
denoted by $\rle(T)$,
is a compressed representation of $T$
that encodes every maximal character run $T[i.. i+e-1]$ in $T$ by $c^e$ if
(1) $T[j] = c$ for all $j$ with $i \le j \le i+e-1$,
(2) $i = 1$ or $T[i-1] \ne c$, and
(3) $i+e-1 = n$ or $T[i+e] \ne c$.
We simply call a maximal character run in $T$ a run in $T$.
Also, we call the number $e$ of characters $c$ in a run $c^e$ the exponent of the run.
The RLE size of string $T$, denoted by $r(T)$, is the number of runs in $T$.
For each $i$ with $1 \le i \le r(T)$,
we denote by $R_i$ the $i$th run of $\rle(T)$.
Also, we denote by $\Beg{i}$~(resp.,~$\End{i}$) the beginning~(resp.,~the ending) position of $R_i$,
and by $\Exp{i}$ the exponent of $R_i$.
Furthermore, we denote by $\chara{i}$ the character that constitutes $R_i$.
A factor of $T$ is said to be \emph{RLE-bounded}
if the factor starts at $\Beg{i}$ and ends at $\End{j}$ for some $i, j$ with $i \le j$.
In what follows, we use $m$ to denote the RLE size of the given string $T$.
\begin{example}
  The RLE of string $T = \mathtt{aaabbcccccabbbb}$ is
  $\texttt{a}^3\texttt{b}^2\texttt{c}^5\texttt{a}^1\texttt{b}^4$.
  The exponents of the first two runs $\texttt{a}^3$ and $\texttt{b}^2$ are three and two, respectively.
  The factor $T[4.. 11] = \texttt{bbccccca}$ of $T$ is RLE-bounded since it starts at $\Beg{2} = 4$ and ends at $\End{4} = 11$.
The RLE size of $T$ is $5$. \end{example}
The following lemma establishes a significant connection between unbordered strings and RLE strings.
\begin{lemma}\label{lem:unbordered}
  Let $u$ be a string of length at least two, and let $a = u[1]$ and $b = u[|u|]$.
  If $u$ is unbordered, then both $au$ and $ub$ are also unbordered.
\end{lemma}
\begin{proof}
  For the sake of contradiction, assume that $au$ is bordered.
  Let $k$ be the length of the border of $au$.
  If $k = 1$, then the border of $au$ is $a$, which implies $a = b$, contradicting the assumption that $u$ is unbordered.
  If $k > 1$, let $x$ be the border of $au$.
  Then, $x[2.. k]$ is a border of $u$, a contradiction.
  Therefore, $au$ must be unbordered.
  The proof for $ub$ is symmetric.
\end{proof}
From Lemma~\ref{lem:unbordered}, any longest unbordered factor of a non-unary string $T$ must be RLE-bounded.
Furthermore, the number of occurrences of longest unbordered factors is at most $m-1$
since any factor starting at the beginning position of the rightmost run is bordered.
Also, the upper bound $m-1$ is tight:
For string $(\mathtt{a}^e\mathtt{b}^e)^\frac{m}{2}$ where $e$ is a positive integer,
all the occurrences of $\mathtt{a}^e\mathtt{b}^e$ and $\mathtt{b}^e\mathtt{a}^e$ are
the longest unbordered factors.
Similarly, the next observation holds:
\begin{observation}\label{obs:contract}
  Let $w$ be a non-empty string.
  If $w\cdot w[|w|]$ has a border of RLE size $p$, then $w$ also has a border of RLE size $p$.
\end{observation}

We further provide an observation about RLE sizes of borders.
\begin{observation}\label{obs:borderRLEsize}
  Let $B$ be a border of a string $T$.
  If $B$ is non-unary, $r(B) \ne r(B')$ holds for any other border $B'$ of $T$.
\end{observation}
\begin{proof}
  Assume that $B$ is longer than $B'$ without loss of generality.
  Then, $B$ is a proper prefix of $T$ and $B'$ is a proper prefix of $B$ since they are borders of $T$.
  Let the last run of $T$ is $b^t$.
  Then, $B = R_1 \cdots R_{r(B)-1} b^t$ holds since $B$ is a border of $T$.
  Let $\tilde{B}$ be a proper prefix of $B$ satisfying $r(\tilde{B}) = r(B)$.
  The last run of $\rle(\tilde{B})$ is $b^{t'}$ with $t' < t$, which implies that $\tilde{B}$ cannot be a suffix of $T$.
  Thus, the RLE size of $B'$, which is a border of $T$ and is shorter than $B$, must be smaller than $r(B)$.
\end{proof}

\subsection{RLE border array}
We define the \emph{RLE border array} $\RLEB$ of $T$ as follows:
For each $i$ with $1\le i \le m$,
$\RLEB[i]$ stores the RLE size of the longest border of the RLE-bounded prefix $T[1..\End{i}]$ of $T$.
For example, when $T = \mathtt{aaabbbaaaaaccaaaabbbaa}$, $\RLEB = [1, 0, 1, 0, 1, 2, 3]$.
\begin{lemma}\label{lem:lbordlinear}
  Given $\rle(T)$, the RLE border array $\RLEB$ of $T$ can be computed in $O(m)$ time.
\end{lemma}
\begin{proof}
  Let $S_i = T[1..\End{i}]$.
  We define an auxiliary array $\RLEBsub$ of size $m$ such that, for each $i$ with $1\le i \le m$, $\RLEBsub[i]$ stores
  the RLE size of the longest border $b$ of $S_i$
  that satisfies the following Condition~($\alpha$) if it exists:
  \begin{center}
    {\bf Condition~($\alpha$):} The prefix-occurrence of $b$ is RLE-bounded.
  \end{center}
  Otherwise, $\RLEBsub[i] := 0$.
  The first two elements of $\RLEB$ and $\RLEBsub$ can be easily obtained.
  $\RLEB[1] = 0$ if $\Exp{1} = 1$, and $\RLEB[1] = 1$ otherwise.
  Also, $\RLEB[2] = \RLEBsub[1] = \RLEBsub[2] = 0$.
  In the rest of the proof, suppose $i \ge 3$ and 
  $\RLEB[i']$ and $\RLEBsub[i']$ have been computed for all $1 \le i' \le i-1$.

  Before explaining the algorithm,
  we consider a relation between borders of $S_i$ and borders of $S_{i-1}$ that satisfy Condition~$(\alpha)$,
  which is the basis of the correctness of our algorithm.
  Let $B$ be a border of $S_i$ with $r(B) \ge 2$.
  The last run of $B$ is identical to $R_i$ since $B$ is a border of $S_i$ and $B$ has at least two runs.
  Then, the string $b$ obtained by trimming the last run from $B$ is a border of $S_{i-1}$.
  Also, the prefix-occurrence of $b$ in $S_{i-1}$ is RLE-bounded, namely, it satisfies Condition~$(\alpha)$.
  By definition, the RLE size of $b$ is $r(B) - 1$.
  Conversely, 
  let $b'$ be a border of $S_{i-1}$ that satisfies Condition~$(\alpha)$.
  Let $h = r(b')+1$.
  Then, the runs immediately after the prefix-/suffix-occurrences of $b'$ are $R_h$ and $R_i$, respectively.
  Thus, border $b'$ can be extended by $\Exp{i}$ characters to the right if and only if
  $\chara{h}=\chara{i}$ and $\Exp{h} \ge \Exp{i}$.
  If the condition holds, the string $b'\chara{i}^{\Exp{i}}$ is a border of $S_i$ of RLE size $h$.
  In addition, this new border satisfies Condition~$(\alpha)$ if and only if $\Exp{h} = \Exp{i}$.
  Therefore, any border $b$ of $S_i$ with $r(b) \ge 2$ can be obtained by extending a border
  of $S_{i-1}$ that satisfies Condition~$(\alpha)$ as above.

  We now describe an algorithm to compute $\RLEB[i]$ and $\RLEBsub[i]$ for $i \ge 3$.
  Using $\RLEBsub[1.. i-1]$, we consider the borders of $S_{i-1}$ that satisfy Condition~($\alpha$) in decreasing order of their lengths,
  and check whether each border can be extended to the right by $R_i$.
  This check can be performed with a single run-comparison.
  In the following, let $b$ be the border of $S_{i-1}$ satisfying Condition~($\alpha$) we are currently focusing on. 
  Let $b_{\mathrm{pre}}$ and $b_{\mathrm{suf}}$ denote the prefix-occurrence and the suffix-occurrence of $b$, respectively.
  Additionally, let $R_h$ be the run immediately following $b_{\mathrm{pre}}$.
  To determine whether $b$ can be extended,
  we compare $R_h$ with $R_i$, leading to the following three cases.
  \begin{description}
    \item[\bf Case 1. $\chara{h} \neq \chara{i}$ or $\Exp{h} < \Exp{i}$:]
      The current border $b$ cannot be extended,
      i.e., $b_{\mathrm{suf}}R_i$ does not occur as a prefix of $T$.
      We continue to the next longest border that satisfies Condition~($\alpha$).
    \item[\bf Case 2. $\chara{h} = \chara{i}$ and $\Exp{h} > \Exp{i}$:]
      In this case, $R_i$ is a proper prefix of $R_h$.
      Hence, $b_{\mathrm{pre}}\chara{h}^{\Exp{i}} = b_{\mathrm{suf}}R_i$ is a border of $S_{i}$.
      If $\RLEB[i]$ is not yet assigned, we set $\RLEB[i] = h$.
      Note that $\RLEBsub[i]$ is not yet determined because $b_{\mathrm{pre}}\chara{h}^{\Exp{i}}$ is not RLE-bounded.
      Thus, we continue to the next longest border that satisfies Condition~($\alpha$).
    \item[\bf Case 3. $\chara{h} = \chara{i}$ and $\Exp{h} = \Exp{i}$:]
      In this case, $R_h = R_i$.
      Consequently, $b_{\mathrm{pre}}R_h=b_{\mathrm{suf}}R_i$ is a border of $S_{i}$ that satisfies Condition~($\alpha$).
      If $\RLEB[i]$ is not yet assigned, we set $\RLEB[i] = h$. In addition, we set $\RLEBsub[i] = h$.
      We terminate the iteration for $i$, and proceed to the computation for $i+1$.
  \end{description}
  As a corner case, if $\RLEBsub[i]$ remains unassigned after checking all borders of $S_{i-1}$ that satisfy Condition~($\alpha$),
  we compare $R_i$ with the first run $R_1$.
  We set $\RLEBsub[i] = 1$ if $\chara{1} = \chara{i}$ and $\Exp{1}\leq\Exp{i}$,
  and otherwise, $\RLEBsub[i] = 0$.
  Furthermore, if $\RLEB[i]$ also remains unassigned, we set
  $\RLEB[i] = 1$ if $\chara{1} = \chara{i}$, and otherwise, $\RLEB[i] = 0$.
  The procedures are performed for each $i = 3, 4, \ldots, m$, after which we obtain the complete array $\RLEB[1.. m]$.

  Clearly, the space usage is $O(m)$.
  The time complexity is linear in the number of run-comparisons.
  Let $R_k$ denote the run containing the first character of $b_{\mathrm{suf}}$.
  Consider the changes in values of $k$ and $i$.
  In short, the number of run-comparisons can be charged to the increasing positive potential $\Phi = k + i \le 2m$.
  For each $i$, the algorithm checks borders that satisfy Condition~($\alpha$) until $\RLEBsub[i]$ is determined.
  Whenever $\RLEBsub[i]$ is not determined following a run-comparison~(i.e., Cases~1 or 2),
  the index $k$ increases by at least one (cf.~Observation~\ref{obs:borderRLEsize}), while $i$ remains unchanged.
  Conversely, when $\RLEBsub[i]$ is determined (i.e.,~Case 3),
  $k$ is unchanged and $i$ increases by one.
  Thus, both $k$ and $i$ are non-decreasing, and at least one of them must increase after each run-comparison.
  Since $1 \le k, i \le m$, the number of run-comparisons is $O(m)$.
\end{proof}

Since any longest unbordered factor of a string is RLE-bounded,
we can compute the longest unbordered factors starting at $\Beg{j}$ in $O(m)$ time
by computing $\RLEB[\Beg{j}.. n]$ for every $j$ with $1 \le j \le m$.
Thus, we have the following corollary:
\begin{corollary}\label{cor:msquaredtime}
  Given $\rle(T)$, we can compute the longest unbordered factors of $T$
  in $O(m^2)$ time using $O(m)$ space.
\end{corollary}

\subsection{Some functions for RLE strings}

Given $\rle(T)$, we construct array $\mathsf{ExpSum}$ of size $m$ such that
$\mathsf{ExpSum}[j] = \sum_{k = 1}^{j}\Exp{k}$ for each $1 \le j \le m$.
Then, given a text position $i$ with $1 \le i \le n$,
we can compute the run to which the $i$th character $T[i]$ belongs
in $O(\log m)$ time by performing binary search on $\mathsf{ExpSum}$.
Tamakoshi et al.~\cite{TamakoshiGIBT15} proposed an $O(m)$-space data structure based on RLE,
called a \emph{truncated RLE suffix array} (tRLESA).
They showed that tRLESA enhanced with some additional information of size $O(m)$
supports several standard string queries, such as pattern matching.
The \emph{truncated RLE suffix array} $\tRLESA$ of $T$ is the integer array such that
$\tRLESA[r] = k$ iff the $r$th lexicographically smallest suffix among
$\mathcal{S} = \{T[\End{x}.. n] \mid 1 \le x \le m\}$ is $T[\End{k}.. n]$.
For convenience, we denote by $T\langle r \rangle = T[\End{\tRLESA[r]}.. n]$
the $r$th lexicographically smallest suffix among $\mathcal{S}$.
The \emph{truncated RLE inverse suffix array} $\tRLEISA$ of $T$ is the integer array
such that $\tRLEISA[k] = r$ iff $\tRLESA[r] = k$.
The \emph{truncated RLE longest common prefix array} $\tRLELCP$ of $T$ is the integer array
such that $\tRLELCP[1] = -1$ and $\tRLELCP[r] = \lcp(T\langle r \rangle, T\langle r-1 \rangle)$ for $2 \le r \le m$.
As in the suffix arrays of plain strings,
we can define the \emph{tRLESA interval} for a factor of $T$ as follows.
For a factor $F$ of $T$,
the tRLESA interval for $F$ is the maximal interval $[r_1, r_2] \subseteq [1, m]$ such that
$F$ is a prefix of $T\langle{r}\rangle$ for every $r \in [r_1, r_2]$.

By applying tRLESA and related data structures in conjunction with RMQ and/or RmQ~(Lemma~\ref{lem:rmq}),
several additional queries can be efficiently supported, as detailed below.
For positive integers $i, p, j, q$ satisfying $i\le m$, $p \le \Exp{i}$, $j \le m$, and $q \le \Exp{j}$,
let $\RLElcp{i}{p}{j}{q}$ denote the length of the longest common prefix of $T[\Beg{i}+p-1..n]$ and $T[\Beg{j}+q-1..n]$.

\begin{lemma}\label{lem:rlefuncA}
  After $O(m\log m)$-time and $O(m)$-space preprocessing for $\rle(T)$,
  the value $\RLElcp{i}{p}{j}{q}$ can be computed in $O(1)$ time
  for given integers $i, p, j$, and $q$.
\end{lemma}
\begin{proof}
  Suppose that an RLE-compressed string $\rle(T)$ of RLE size $m$ is given.
  In the preprocessing phase, we construct length-$m$ arrays $\tRLESA$, $\tRLEISA$, and $\tRLELCP$ of string $T$. We further construct an RmQ (range \emph{minimum} query) data structure on $\tRLELCP$.
  Given query integers $i, p, j$, and $q$,
  we first check the equality of the characters.
  If $\chara{i} \ne \chara{j}$, we immediately return $0$. Otherwise,
  we compare $\Exp{i}-p+1$ and $\Exp{j}-q+1$.
  If they are not equal, the lcp value is the minimum of those values.
  Otherwise, let $r_1 =\min\{\tRLEISA[i], \tRLEISA[j]\}$ and $r_2 = \max\{\tRLEISA[i], \tRLEISA[j]\}$.
  Then, the lcp value to return is
  $\Exp{i}-p+\min\{\tRLELCP[r] \mid r_1 + 1 \le r \le r_2\}$,
  which can be computed in $O(1)$ time by answering RmQ on $\tRLELCP$.
\end{proof}

For four integers $x, y, h$, and $\ell$ such that
$1 \le x \le y \le m$,
$1 \le h \le m$, and 
$0 \le \ell \le n+1$,
we define
\[
  \LongestPref{x}{y}{h}{\ell}=
  \begin{cases}
    -1, & \{z\in[x,y]\mid L_z\le\ell\}=\emptyset,\\
    \displaystyle
    \operatorname*{argmax}_{z\in[x,y]} \{L_z \mid L_z \le \ell\}, & \text{otherwise},
  \end{cases}
\]
where $L_z=\lcp(T[\End{h}..n],T[\End{z}.. \End{y}])$.

In words, $\LongestPref{x}{y}{h}{\ell}$ is 
the index $z \in [x, y]$ of a run such that
$T[\End{z}.. \End{y}]$ has the longest lcp with $T[\End{h}..n]$
where the lcp length is at most $\ell$, if such $z$ exists.
\begin{lemma}\label{lem:rlefuncB}
  After $O(m\log m)$-time and $O(m)$-space preprocessing for $\rle(T)$ and fixed integers $x, y$ with $1 \le  x \le y \le m$,
  the value $\LongestPref{x}{y}{h}{\ell}$ can be computed in $O(\log m)$ time
  for given integer pair $(h, \ell)$ with $0 \le \ell \le n+1$ and $1\le h\le m$.
\end{lemma}
\begin{proof}
  Given $\rle(T)$ of RLE size $m$ and integers $x, y$, we consider string $T' = T\$T[\Beg{x}.. \End{y}]$.
  If a string occurs within $T[\Beg{x}.. \End{y}]$, then it also occurs at some position $p > n+1$ in $T'$.
  Further, the index of the run containing such a position $p > n+1$ is greater than or equal to $m+2$.
  Thus, by filtering the run indices using the threshold $m+2$,
  we can extract suffixes of $T[\Beg{x}.. \End{y}]$ that start at the end of a run.

  In the preprocessing, we construct three arrays for $T'$ and an RmQ data structure on $\tRLELCP$
  as in the proof of Lemma~\ref{lem:rlefuncA}.
  In addition, we construct an RMQ (range \emph{maximum} query) data structure on $\tRLESA$.

  Given query integers $h$ and $\ell$, we first obtain $r_h=\tRLEISA[h]$.
  We then compute the maximal tRLESA interval $[r_1,r_2]$ containing $r_h$ such that
  $\lcp(T'\langle r\rangle,T'\langle r_h\rangle)\ge\ell+1$ for every $r\in[r_1,r_2]\setminus\{r_h\}$.
  Thus, every suffix satisfying the lcp constraint lies outside $[r_1, r_2]$.
  As a special case, if $\ell = n+1$, we just set $[r_1, r_2]=[r_h, r_h]$.
  For the fixed query suffix $T'\langle r_h\rangle$,
  the lcp value is non-decreasing as the lexicographical rank approaches $r_h$ from either side.
  Hence, it suffices to examine $r_3$ and $r_4$ defined as follows:
  $r_3 = \max\{r \in [1, r_1-1]     \mid  \tRLESA[r] \ge m+2\}$; and
  $r_4 = \min\{r \in [r_2+1, r(T')] \mid  \tRLESA[r] \ge m+2\}$;
  where either value is undefined if the corresponding set is empty.
  Finally, if at least one of $r_3$ and $r_4$ is defined, we return $\tRLESA[r^\star] - (m+2) + x$ where
  $r^\star = \operatorname{argmax}_{r \in \{r_3, r_4\}}\{\lcp(T'\langle r \rangle, T'\langle r_h \rangle)\}$.
  Otherwise, we return $-1$.

  The range $[r_1, r_2]$ can be computed in $O(\log m)$ time
  by combining binary search and RmQ on $\tRLELCP$.
  The values $r_3$ and $r_4$ can also be computed in $O(\log m)$ time
  by combining binary search and RMQ on $\tRLESA$ over $[1,r_1-1]$ and $[r_2+1,r(T')]$, respectively.
\end{proof}
\section{\texorpdfstring{$O(m\sqrt{m}\log^2 m)$}{O(m1.5log2m)}-time algorithm}\label{sec:algo}

In this section, we prove our main theorem:
\begin{restatable}{theorem}{main}\label{thm:main}
  Given an RLE encoded string $\rle(T)$ of RLE size $m$,
  we can compute the set of longest unbordered factors of $T$
  in $O(m\sqrt{m}\log^2 m)$ time using $O(m \log^2 m)$ space.
\end{restatable}

To simplify notation, we present Sections~\ref{sec:algo} and~\ref{sec:mlog2m} under the assumption that $m$ is a power of four.
For an arbitrary RLE size $m$, the same algorithms use
$M=4^{\lceil\log_4 m\rceil}$ only as the scale parameter in block sizes and run-count thresholds;
the input string itself still has exactly $m$ runs,
and all index ranges are restricted to $[1, m]$.
Since $m \le M < 4m$, this convention does not affect the correctness
or the stated asymptotic bounds.
We suppress $M$ and the associated rounding symbols below.

The high-level strategy of our algorithm, presented in Algorithm~\ref{alg:main}, is essentially the same as that of Gawrychowski et al.~\cite{GawrychowskiKSS15}.
We divide the input string $T$ into $K := \lceil m / \sqrt{m} \rceil$ blocks $J_1, J_2, \ldots, J_{K}$,
where each block $J_k$ is RLE-bounded and has RLE size $\sqrt{m}$ for every $1 \leq k < K$.
The RLE size of the last block $J_K$ may be smaller than $\sqrt{m}$.
Throughout this section, we refer to a border of RLE size at most $\sqrt{m}$ as a \emph{short} border, and otherwise as a \emph{long} border.
Let $\rho_T$ denote the RLE size of a longest unbordered factor of $T$.
Also, we refer to a factor starting at $\Beg{i}$ and ending within the $k$th block $J_k$ as an \emph{$(i, k)$-factor}.
Note that the longest unbordered factor of $T$ must be an $(i, k)$-factor for some $i$ and $k$
since it is RLE-bounded~(see Lemma~\ref{lem:unbordered}).

\begin{algorithm}[tb]
  \caption{Algorithm for computing the longest unbordered factor}
  \label{alg:main}
  \begin{algorithmic}[1]
    \Require String $T$ of RLE size $m$.
    \Ensure The set of longest unbordered factors of $T$.
    \State $\LUB \gets \LongSUBf(T)$ \Comment{$\LUB$: set of longest unbordered factors.}
    \State $\ell^\star \leftarrow \max\{ |x|\,:\, x \in \LUB\}$ \Comment{$\lub$: length of the longest unbordered factor.}
    \State Preprocess for {\RmBordf}
    \For{$k \gets 5$ to $K$}
    \State Preprocess for $\Candf_k$
    \State $C \gets \{\varepsilon\}$
    \For{$i \gets 1$ to $(k-4)\sqrt{m}$}
    \State $C \gets C \cup \{\Candf_k(i)\}$
    \EndFor
    \State $U \gets \RmBordf(k, C)$ \Comment{All strings in $U$ are unbordered.}
    \State $\ell \gets \max\{|u|\,:\,u\in U\}$
    \If{$\ell < \lub$}
    \State \textbf{continue} \Comment{Do nothing and continue to the next stage.}
    \ElsIf{$\ell > \lub$}
    \State $\LUB \gets \emptyset$ \Comment{Clear the current tentative solutions.}
    \State $\lub \gets \ell$
    \EndIf
    \State $\LUB \gets \LUB \cup \{u \in U\,:\, |u| = \lub\}$
    \EndFor
    \State  \Return $\LUB$
  \end{algorithmic}
\end{algorithm}

Let us look at Algorithm~\ref{alg:main}.
The set $\LUB$ represents the current tentative solution and
the variable $\lub$ represents the length of an element in $\LUB$.
First of all, we invoke the subroutine {\LongSUBf},
which outputs the longest unbordered factors of RLE size at most $4\sqrt{m}$,
and tentatively updates $\LUB$ and $\ell^\star$~(lines~1--2)\footnote{
  The constant factor $4$ follows the choice of Gawrychowski et al.~\cite{GawrychowskiKSS15}
  and is used only for convenience; we do not attempt to optimize it.
}.
The main part of our algorithm consists of $O(\sqrt{m})$ stages,
corresponding to the outer \textbf{for} loop (lines~4--19).
In each stage, say the $k$th stage, we first compute a set $C$ of \emph{candidates} for
longest unbordered factors that end within $J_k$
by calling the subroutine $\Candf_k$ $O(m)$ times (lines~7--9).
Here, as we will show later in Lemma~\ref{lem:cand},
the subroutine $\Candf_k(i)$ returns one of the following three strings:
(1) the longest unbordered $(i, k)$-factor, if such a factor exists;
(2) the empty string, if all $(i, k)$-factors have short borders; or
(3) an $(i, k)$-factor that has no short border but has a long border, otherwise.
Then, the set $C$ is guaranteed to contain a longest unbordered factor of $T$
if (i) $\rho_T > 4\sqrt{m}$ and (ii) there is a longest unbordered factor ending within $J_k$.
We then eliminate from $C$ all factors that have a long border by calling the subroutine $\RmBordf(k, C)$,
which checks whether each string in $C$ has a long border and removes it if so (line~10).
If any candidates remain, we select the longest ones and update $\lub$ and $\LUB$ if necessary (lines~11--18).
After the outer \textbf{for} loop,
we have the final answer $\LUB$, thus output it.

The correctness of Algorithm~\ref{alg:main} is straightforward from the properties of the three subroutines
{\LongSUBf}, $\Candf_k$, and {\RmBordf}.
In what follows, we describe how to implement the subroutines efficiently.

\subsection{Implementation of LONGEST-SHORT-UB}
To implement $\LongSUBf$, we simply apply Lemma~\ref{lem:lbordlinear} for all
RLE-bounded factors of RLE size $4\sqrt{m}$.
By doing so, we can compute the longest unbordered factors
of RLE size at most $4\sqrt{m}$ in a total of $O(m\sqrt{m})$ time.

\subsection{Implementation of \texorpdfstring{CANDIDATE$_k$}{CANDIDATEk}}\label{sec:candidate}

Throughout this subsection, we fix an integer $5 \le k \le K$ arbitrarily.
The definition of function $\Candf_k(i)$ is as follows:
$\Candf_k(i)$ returns the longest $(i, k)$-factor that has no short border, if it exists; or
the empty string, otherwise.
This definition leads to the following properties:
\begin{lemma}\label{lem:cand}
  (1) If there is an unbordered $(i, k)$-factor, then {\sl $\Candf_k(i)$} returns the longest unbordered $(i, k)$-factor.
  (2) If all $(i, k)$-factors have short borders,
  then {\sl $\Candf_k(i)$} returns the empty string.
  (3) Otherwise, {\sl $\Candf_k(i)$} returns an $(i, k)$-factor that has no short border but has a long border.
\end{lemma}
\begin{proof}
  (1) Consider the case where there is an unbordered $(i, k)$-factor,
  and let $U$ be the longest unbordered $(i,k)$-factor.
  For the sake of contradiction, assume that $U$ is shorter than $\Candf_k(i)$.
  Then, $\Candf_k(i)$ has a long border since it has no short border.
  However, the suffix-occurrence of the long border starts before $J_k$ since $r(J_k) = \sqrt{m}$,
  which contradicts that $U$ is unbordered.
  Thus, $\Candf_k(i)$ is the longest unbordered $(i,k)$-factor in this case.
  (2) If all $(i, k)$-factors have short borders, then $\Candf_k(i)$ returns $\varepsilon$ by definition.
  (3) Otherwise, i.e., all $(i, k)$-factors have some borders.
  By definition, $\Candf_k(i)$ has no short border, thus it must have a long border.
\end{proof}

Let $J_k.\first$ and $J_k.\last$ be the indices of runs such that
$T[\Beg{J_k.\first}..$
$\End{J_k.\last}] = J_k$.
Let $D_k = J_{k-1}J_k$, $x = J_{k-1}.\first$, $y = J_k.\first$, and $z = J_k.\last$.
Namely, $T[\Beg{x}.. \End{z}] = D_k$ and $T[\Beg{y}.. \End{z}] = J_k$ hold.
Let $P_i$ be the longest prefix of $T[\End{i}.. \End{z}]$ that occurs in $D_k$.
If $P_i = \varepsilon$, then $\Candf_k(i)$ returns $T[\Beg{i}.. \End{z}]$ since it has no short border.
Otherwise, let $p = r(P_i)$.
If $p = 1$, $\Candf_k(i)$ can be easily computed by comparing the characters of the $i$th run and the $z$th run.
Thus, we assume $p > 1$ in the following.
Let $P_i = aub^{e_1}$ where $a = P_i[1]$, $b^{e_1}$ is the last run of $P_i$, and $u \in \Sigma^\star$ is the rest.
Further let $\Gamma$ be the set of exponents of $b$ following some occurrences of $au$ in $D_k$.
If $\min\Gamma > e_1$, the next character of $aub^{e_1}$ in $D_k$ is always $b$.
Then, any factor starting at $\End{i}$ and ending at $\End{j'}$ for some $y \le j' \le z$
can not have a short border of RLE size exactly $p$
because if such a border exists, the border forms $aub^e$ with $e \le e_1$, contradicting $\min\Gamma > e_1$.
If $\min\Gamma \le e_1$, we define
$e_2=\max\{\gamma\in\Gamma\mid\gamma \le e_1\}$.
Let $\End{t}$ be the starting position of an occurrence of $aub^{e_2}$ in $D_k$.
We further define
$F(t, j) = T[\End{t}.. \End{z}]\$T[\Beg{x}.. \End{y+j-1}]$ for $j$ with $1 \le j \le \sqrt{m}$,
where $\$$ is a special character with $\$\not\in\Sigma$.
The next lemma holds:
\begin{lemma}\label{lem:RLEsizep}
  Assume $\min\Gamma \le e_1$ holds.
  For each $1 \le j \le \sqrt{m}$,
  $T[\End{i}.. \End{y+j-1}]$ has a short border of RLE size $p$ if and only if
  $F(t, j)$ has a short border of RLE size $p$.
\end{lemma}
\begin{proof}
  Let $j' = y+j-1$.
  ($\Rightarrow$)
  If $T[\End{i}.. \End{j'}]$ has a short border of RLE size $p \le \sqrt{m}$,
  the border forms $aub^{\Exp{j'}}$ and it holds that $\Exp{j'} \le e_1$.
  Also, $\Exp{j'} \le e_2$ holds since $e_2 < \Exp{j'} \le e_1$ contradicts the definition of $e_2$.
  Therefore, $F(t, j)$ has border $aub^{\Exp{j'}}$ since $F(t, j)$ starts with $aub^{e_2}$.
  ($\Leftarrow$)
  If $F(t, j)$ has a short border of RLE size $p \le \sqrt{m}$, 
  the border forms $aub^{\Exp{j'}}$ and it holds that $\Exp{j'} \le e_2$.
  Also, $\Exp{j'} \le e_2 \le e_1$ holds by the definition of $e_2$.
  Therefore, $T[\End{i}.. \End{j'}]$ has border $aub^{\Exp{j'}}$ since $T[\End{i}.. \End{j'}]$ starts with $aub^{e_1}$.
\end{proof}

Let $t^\star = t$ if $\min\Gamma \le e_1$;
otherwise,
let $t^\star$ be the run index such that $\End{t^\star}$ is the starting position of an occurrence of $aub$ in $D_k$.
Regardless of the value of $\min\Gamma$,
for each $1 \le j \le \sqrt{m}$,
the set of short borders of $T[\End{i}.. \End{y+j-1}]$ of RLE size $q$
is equivalent to that of $F(t^\star, j)$ for any $q < p$,
since such borders are prefixes of their common prefix $au$.
Also, the RLE size of any short border of $T[\End{i}.. \End{y+j-1}]$ is
upper bounded by $p$ from the definition of $P_i$.
To summarize, 
the set of short borders of $T[\End{i}.. \End{y+j-1}]$ is equivalent to
the set of short borders of $F(t^\star, j)$ of RLE size at most $p-1$ if $\min\Gamma > e_1$;
otherwise, it is equivalent to
the set of short borders of $F(t^\star, j)$ of RLE size at most $p$ by Lemma~\ref{lem:RLEsizep}.

Next, for $1 \le j \le \sqrt{m}$, let us consider the short borders of $T[\Beg{i}.. \End{j'}]$ where $j' = y + j - 1$.
A short border of $T[\Beg{i}.. \End{j'}]$ can be obtained
by extending a short border of $T[\End{i}.. \End{j'}]$ to the left by $\Exp{i}-1$ characters.
Consider all the short borders $B_1, B_2, \ldots, B_g$ of $T[\End{i}.. \End{j'}]$,
which satisfy that  $r(B_s) \le \min\{p, \sqrt{m}\}$ for all $s$.
Note that all such borders are also borders of $F(t^\star, j)$ as discussed above.
Let $e_s$ be the exponent of the first run of the suffix-occurrence of $B_s$ for each $1 \le s \le g$.
Let $\mathcal{E}_{j'}^p$ be the set of such $e_s$ for all $B_1, B_2, \ldots, B_g$.
If $e_s < \Exp{i}$ for all $e_s \in \mathcal{E}_{j'}^p$, i.e., $\max\mathcal{E}_{j'}^p < \Exp{i}$,
then $T[\Beg{i}.. \End{j'}]$ cannot have a short border.
Conversely, if $\max\mathcal{E}_{j'}^p \ge \Exp{i}$,
then $T[\Beg{i}.. \End{j'}]$ has a short border.

Based on the observations above, we design an algorithm for computing $\Candf_k(i)$.

\paragraph*{\bf Preprocessing.}
We construct a data structure for $\LongestPref{x}{z}{\cdot}{\cdot}$ by using Lemma~\ref{lem:rlefuncB}.
Next, let us \emph{conceptually} consider a $\sqrt{m}\times\sqrt{m}$ table $M_\tau$ defined as follows:
$M_\tau[r][j] = \infty$ for all $r$
if the first and the last characters of $F(\tau, j)$ are the same;
otherwise,
$M_\tau[r][j]$ stores the maximum exponent among the first runs of the suffix-occurrences of
those borders of $F(\tau, j)$ whose RLE size is at most $r$;
if there is no such border, then $M_\tau[r][j] = 0$~(see also the left part of Fig.~\ref{fig:segments}).
Note that we do not explicitly construct such $\sqrt{m}\times\sqrt{m}$ tables.
The details of their implementation will be described later.
Below, we use $n+1$ as a finite representation of $\infty$ since every lcp value is at most $n$.

\paragraph*{\bf Query.}
Given a position $i$,
we first compute $\alpha = \LongestPref{x}{z}{i}{\infty}$,
which satisfies that the lcp of $T[\End{i}..n]$ and $T[\End{\alpha}..\End{z}]$ equals $P_i$.
Also, we compute $e_1 = |P_i| - |au|$ and $p = r(P_i)$.
If $\Exp{\alpha+p-1} = e_1$, then we set $t = \alpha$ since $e_2 = e_1$.
Otherwise, we compute $\beta = \LongestPref{x}{z}{\alpha}{|P_i|}$.
Next, we compute
$L := \min\{
  \RLElcp{i}{\Exp{i}}{\beta}{\Exp{\beta}},
  \End{z}-\End{\beta}+1
\}$.
The second term is obtained in $O(1)$ time from $\ExpSum$.
By the maximality of $P_i$, this capped value satisfies $L\le |P_i|$.
If $L\le |au|$, then $\min\Gamma>e_1$ and we set $t=\alpha$.
Otherwise, we set $t=\beta$; in this case $e_2=L-|au|\le |P_i|-|au|=e_1$.
Furthermore, we set $p' = \min\{p-1, \sqrt{m}\}$ if $\min\Gamma > e_1$; otherwise, set $p' = \min\{p, \sqrt{m}\}$.
Next, we find the largest $j^\star$ such that $M_t[p'][j^\star] < \Exp{i}$.
If there is no such $j^\star$, $\Candf_k(i)$ returns $\varepsilon$.
Otherwise, it returns $T[\Beg{i}.. \End{y+j^\star-1}]$ since
$T[\Beg{i}.. \End{\iota}]$ has a short border for all $\iota$ with $y+j^\star-1 < \iota \le z$,
and hence, by Observation~\ref{obs:contract},
$T[\Beg{i}.. q]$ has a short border for all text positions $q$ with $\Beg{y+j^\star} \le  q \le \End{z}$.

\subsubsection{Implementing table \texorpdfstring{$M_\tau$}{M\_tau}}

The remaining task is to efficiently implement $M_\tau$ so that $j^\star$ can be found quickly.

We first show that each column consists of only $O(\log m)$ maximal intervals of equal values.
For each column of $M_\tau$, say the $j$th column, the values are non-decreasing.
If the first and last characters of $F(\tau,j)$ are identical, then $M_\tau[r][j]=\infty$ for every $r$, and hence the column consists of a single interval of equal values.
In what follows, suppose that these characters are distinct.
Consider the string
$W=T[\End{\tau}+1..\End{z}]\$T[\Beg{x}..\End{y+j-2}]$,
which is obtained by removing the first and last runs of $F(\tau,j)$.
For every border $B$ of $F(\tau,j)$ with $r(B)\geq 3$,
removing the first and last runs of $B$ gives an RLE-bounded border of $W$,
and hence a border of $\rle(W)$, regarded as a string over $\Sigma\times\mathbb{N}$.
Below we call an element of $\Sigma\times\mathbb{N}$ a \emph{run symbol}.
We next explain a property of borders within a single \emph{border group},
which is a set of borders having the same smallest period
(cf.~Lemma~\ref{lem:periodofborder}).
Let $b_0>b_1>\cdots>b_t$ be the lengths of borders of $\rle(W)$
belonging to one border group,
and let $p$ be their common smallest period.
By Lemma~\ref{lem:periodofborder}, they form a decreasing arithmetic progression satisfying
$b_{h-1}-b_h=p$ for every $1\leq h\leq t$;
equivalently, $b_h=b_0-hp$.
Let $B_0$ be the border of $\rle(W)$ of length $b_0$.
Since $B_0$ has period $p$,
the run symbols immediately following the prefix-occurrences of the non-longest borders are identical:
they all lie inside the prefix-occurrence of $B_0$, and their positions in $\rle(W)$ differ by multiples of $p$.
Symmetrically, the run symbols immediately preceding their suffix-occurrences are identical,
since they all lie inside the suffix-occurrence of $B_0$ and their positions differ by multiples of $p$.
Thus, all borders in the group except $B_0$ have the same ordered pair
$(c_1^{e_1}, c_2^{e_2})$ of run symbols,
where $c_1^{e_1}$ follows the prefix-occurrence and
$c_2^{e_2}$ precedes the suffix-occurrence.
Consequently, either all non-longest members of the group can be extended to borders of $F(\tau,j)$, or none of them can.
In the former case, they all yield the same exponent $e_2$.
The longest member of the group may contribute one additional exponent value.
Hence, each border group contributes at most two distinct exponent values.
Since the first and last characters of $F(\tau,j)$ are distinct, it has no unary border;
moreover, by Observation~\ref{obs:borderRLEsize}, it has at most one border of RLE size two.
Thus, borders of RLE size at most two contribute only $O(1)$ additional exponent values.
Since there are $O(\log m)$ border groups, each column of $M_\tau$
contains only $O(\log m)$ distinct values and therefore consists of
$O(\log m)$ maximal intervals of equal values.

We define a (vertical) line segment that corresponds to each such interval
and assign its common integer value to the segment as its weight~(see Fig.~\ref{fig:segments}).
\begin{figure}[tb]
  \centering
  \includegraphics[width=\linewidth]{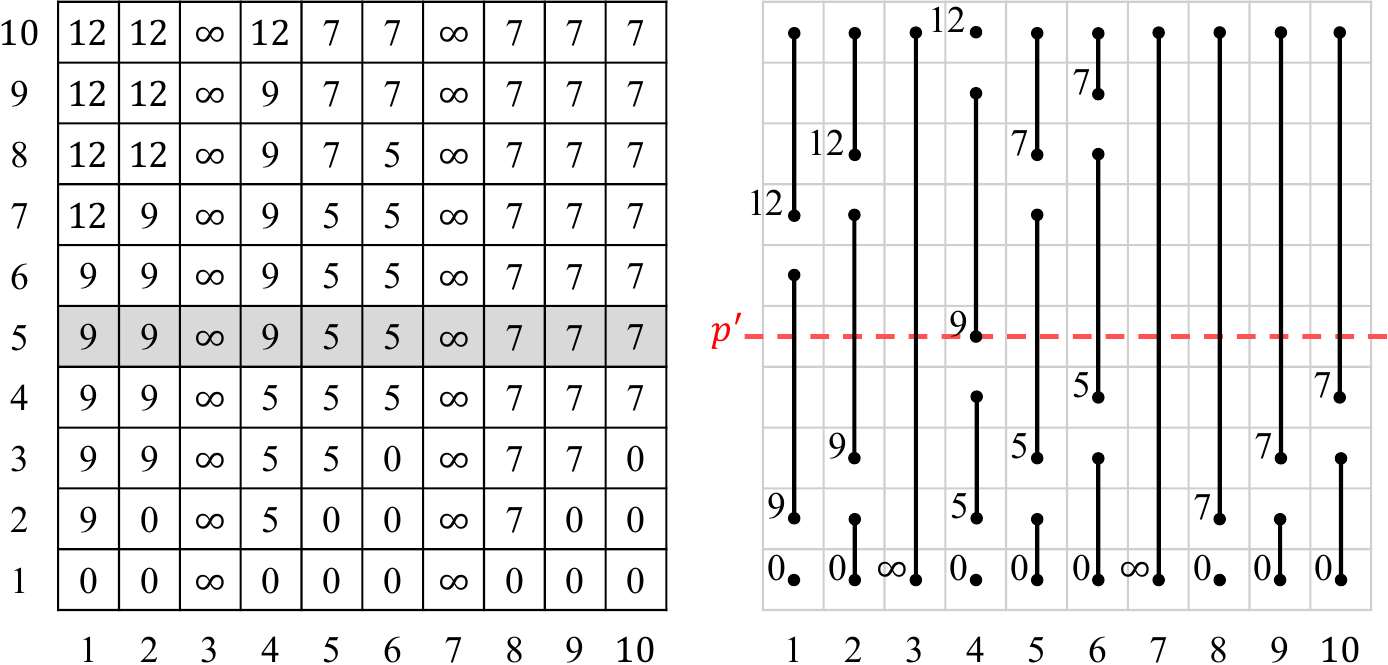} 
  \caption{Left: An example of table $M_\tau$.
    Each column is non-decreasing from bottom to top.
    The $\infty$ symbols in the $3$rd and $7$th columns indicate that the first and last characters of $F(\tau, 3)$ and $F(\tau, 7)$ are the same, respectively.
    Right: Line segments corresponding to the runs in all columns of $M_\tau$.
    If $p' = 5$ and $e = 7$,
    the largest $j^\star$ such that $M_\tau[p'][j^\star] < e$
    is $6$.
  }\label{fig:segments}
\end{figure}
Let $S_\tau$ denote the resulting set of weighted line segments over all columns of $M_\tau$.
Since $M_\tau$ has $O(\sqrt{m})$ columns and each column has $O(\log m)$ segments, we have
$|S_\tau|=O(\sqrt{m}\log m)$.
Using these segments, for given $p'$ and $e$, we can find the largest
$j^\star$ such that $M_\tau[p'][j^\star]<e$ as follows:
we find the rightmost segment that intersects the horizontal line
$r=p'$ and has weight less than $e$; then $j^\star$ is the
$j$-coordinate of this segment.

We next describe how to construct $S_\tau$.
To make the dependence on $j$ explicit, let
$W_j=T[\End{\tau}+1..\End{z}]\$T[\Beg{x}..\End{y+j-2}]$
for every $1\leq j\leq z-y+1$, and let
$W^\star=W_{z-y+1}=T[\End{\tau}+1..\End{z}]\$T[\Beg{x}..\End{z-1}]$.
For fixed $\tau$, we construct the border array and the border-group
array of $\rle(W^\star)$.
Since every $W_j$ ends at a run boundary, $\rle(W_j)$ is a prefix of
$\rle(W^\star)$.
Therefore, these arrays allow us to enumerate the border groups of
$\rle(W_j)$ for every $j$.
For each border group of $\rle(W_j)$, we examine its longest member
and, if the group contains at least two members, its shortest member.
For each examined member, we test in constant time whether it can be
extended by the first and last runs removed from $F(\tau,j)$ to form
a border of $F(\tau,j)$.
If it can, we record the pair $(E,q)$, where $E$ is the exponent
contributed by the border and $q$ is its RLE size.
As shown above, the longest member represents itself, while the
shortest member represents all non-longest members of the group:
all non-longest members yield the same exponent, and the shortest
one has the minimum RLE size among them.
Borders of RLE size at most two are handled separately in constant
time.

This procedure produces $O(\log m)$ pairs for each column.
If the same exponent $E$ occurs in multiple pairs, we retain only the
pair having the minimum RLE size.
We then sort the remaining pairs by their RLE sizes in ascending order
and scan them while maintaining the maximum exponent encountered so
far.
Pairs having the same RLE size are processed together.
Whenever the maintained maximum changes, we start a new segment.
In this way, we obtain exactly the maximal intervals of equal values
in the $j$th column.

The border array and the border-group array of $\rle(W^\star)$ can be
constructed in $O(\sqrt{m})$ time.
For each of the $O(\sqrt{m})$ columns, enumerating the border groups
produces $O(\log m)$ pairs, and sorting these pairs takes
$O(\log m\log\log m)$ time.
Consequently, $S_\tau$ can be constructed in
$O(\sqrt{m}\log m\log\log m)$ time.

For each $x \le \tau \le z$,
we construct a data structure of the \emph{weighted lowest stabbing query (WLSQ)} on $S_\tau$, where $S_\tau$ is rotated 90 degrees to the right, as defined below:
\begin{definition}[Weighted Lowest Stabbing Query; WLSQ]\label{def:geo}
  A set $S$ of weighted horizontal segments over $N\times N$ grid are given for preprocessing.
  The query is, given integers $v$, $w_1$, and $w_2$,
  to report the lowest segment $s$ such that
  $s$ is stabbed by vertical line $x = v$
  and the weight of $s$ is between $w_1$ and $w_2$.
\end{definition}
Given $i$, we can obtain $j^\star$ by answering WLSQ on $S_t$ for $v = p', w_1 = 0$, and $w_2 = \Exp{i}-1$.

Very recently, Akram and Mieno~\cite{Akram,AkramMienoPreprint} proposed an algorithm for a problem called
the \emph{2D top-$k$ stabbing query with weight constraint}~(Definition 8 in~\cite{Akram}),
which subsumes WLSQ as a special case.
Their data structure can be constructed in $O(|S| \log^2 |S|)$ time~\cite{AkramMienoPreprint}.
We propose a simpler data structure specialized for WLSQ
and show that it can be constructed slightly faster, in $O(|S| \log |S|)$ time:
\begin{lemma}\label{lem:geo}
  A set $S$ of segments is given as the input of WLSQ.
  We can build a data structure of size $O(|S| \log |S|)$
  in $O(|S| \log |S|)$ time
  that can answer any WLSQ in $O(\log |S|)$ time.
\end{lemma}
\begin{proof}
  In the preprocessing phase,
  we first sort the segments by their weights. Then, we construct the segment tree of $S$ with $N$ leaves
  by inserting segments in the sorted order. Each segment stored in a node is represented as a tuple $(w, y, \mathsf{id})$
  where $w$ is the weight of the segment,
  $y$ is the y-coordinate of the segment,
  and $\mathsf{id}$ is the unique id of the segment in $S$.
  The resulting segment tree satisfies the condition that
  segments stored in each node are sorted by their weights.
  Subsequently, for each node of the segment tree,
  we build an RmQ data structure for the list of y-coordinates stored in the node.
  Lastly, we apply the fractional cascading to the series of tuples over the segment tree,
  using weight $w$ as the key.

  Given a query $(v, w_1, w_2)$, we first find the path $\pi$ from the root to the leaf that corresponds to $v$ in the segment tree.
  We then traverse $\pi$ starting from the root.
  At each node $u \in \pi$,
  we compute the range corresponding to $[w_1, w_2]$ in the list of tuples in $u$ by using the fractional cascading structure,
  and compute the smallest y-coordinate within the range by using the RmQ data structure.
  At the end, we output the $\mathsf{id}$ that corresponds to the smallest y-coordinate within the minima.

  The size of the data structure is $O(|S| \log |S|)$, which is dominated by the segment tree.
  The construction time is $O(|S| \log |S|)$:
  sorting the segments takes $O(|S| \log |S|)$ time,
  building the segment tree requires $O(|S| \log |S|)$ time,
  and constructing the RmQ data structures and applying fractional cascading take time linear in the total size of the lists,
  i.e., $O(|S| \log |S|)$.
  Given a query, we first find the root-to-leaf path $\pi$ of length $O(\log |S|)$.
  Detecting all the ranges corresponding to $[w_1, w_2]$ over the lists in $\pi$ can be done in $O(\log |S|)$ time
  thanks to fractional cascading.
  We then perform an $O(1)$-time RmQ $|\pi|$ times,
  which takes $O(\log |S|)$ time.
  Thus, the query time is $O(\log |S|)$ in total.
\end{proof}

Thus, we can implement $\Candf_k$ so that $\Candf_k(i)$ can be computed in $O(\log m)$ time for each $i$
after $O(\sum_{x \le \tau \le z} |S_\tau|\log |S_\tau|)=O(m\log^2 m)$ time and space preprocessing.

\subsection{Implementation of RM-LONG-BORDERED} \label{sec:longbordercheck}

We define a new notion called the \emph{RLE pseudo period} as follows:
for a string $w$ of RLE size $r$, the RLE pseudo period $\pp(w)$ of $w$
is the value $r  - b$ where $b$ is the RLE size of the border of $w$.
Note that the RLE size $p$ of the prefix (or suffix) of $w$ whose length equals the period of $w$
is not always equal to $\pp(w)$,
but it holds that $p-1 \le \pp(w) \le p$.
\begin{example}
  Consider string $w = \mathtt{abaababa}$ of RLE size $r = 7$.
  For this string, $\pp(w) = 4$ holds since the border of $w$ is $\mathtt{aba}$ of RLE size $b = 3$.
  The period of $w$ is $5$.
  On the one hand, the RLE size $p_1$ of the length-$5$ prefix $\mathtt{abaab}$ of $w$ is $4$, and then $\pp(w) = p_1$ holds.
  On the other hand, the RLE size $p_2$ of the length-$5$ suffix $\mathtt{ababa}$ of $w$ is $5$, and then $\pp(w) = p_2-1$ holds.
\end{example}
For each $j > \sqrt{m}$, let $S_j$ be the longer string within
(1) the shortest suffix of $\$T[1..\End{j}]$ whose RLE pseudo period is greater than $\sqrt{m}/2-1$, and
(2) $T[\Beg{j-\sqrt{m}+1}.. \End{j}]$.
Note that $S_j$ is well-defined since the RLE pseudo period of $\$T[1..\End{j}]$ must be greater than $\sqrt{m}/2-1$ for $j > \sqrt{m}$.

\begin{lemma}\label{lem:Sjissuffix}
  If the shortest border of $T[\Beg{i}.. \End{j}]$ is a long border,
  then $S_j$ is a suffix of each border of $T[\Beg{i}.. \End{j}]$.
\end{lemma}
\begin{proof}
  Let $B$ be the shortest border of $T[\Beg{i}.. \End{j}]$.
  Since $B$ is unbordered, the pseudo period of $B$ equals the RLE size of $B$,
  which is greater than $\sqrt{m}$.
  Thus, $|S_j| \le |B|$ holds, and thus $S_j$ is a suffix of $B$.
  Therefore, $S_j$ is a suffix of each border of $T[\Beg{i}.. \End{j}]$.
\end{proof}

\begin{lemma}\label{lem:occofSj}
  The number of occurrences of $S_j$ in $T$ is in $O(\sqrt{m})$.
\end{lemma}
\begin{proof}
  The claim is immediate for $m<4$, so assume $m\ge4$.
  Since $r(S_j)\ge\sqrt{m}\ge2$, $S_j$ is non-unary, and hence at most one occurrence of $S_j$ starts in each run of $T$.
  Consider two consecutive occurrences starting in runs $R_u$ and $R_v$, where $u<v$.
  If they overlap, their overlap is a border of $S_j$ of RLE size $r(S_j)-(v-u)$,
  and therefore $v-u\ge \mathit{pp}(S_j)>\sqrt{m}/2-1$.
  If they do not overlap, then $v-u\ge r(S_j)-1\ge\sqrt{m}-1$.
  Thus, consecutive occurrences start $\Omega(\sqrt{m})$ runs apart.
  Since $T$ has $m$ runs, the number of occurrences is $O(m/\sqrt{m})=O(\sqrt{m})$.
\end{proof}

\begin{lemma}\label{lem:datastrforRM}
  Given $j > \sqrt{m}$, the string $S_j$ and all its occurrences can be computed in $O(\sqrt{m})$ time
  after $O(m \log m)$ time preprocessing.
\end{lemma}
\begin{proof}
  We first compute the RLE pseudo period $\pi$ of
  $T[\Beg{j - \sqrt{m} + 1}.. \End{j}]$ in $O(\sqrt{m})$ time~(Lemma~\ref{lem:lbordlinear}).
  If $\pi > \sqrt{m}/2 -1$, then
  $S_j = T[\Beg{j - \sqrt{m} + 1}.. \End{j}]$.
  Otherwise, we compute the 
  longest suffix $T[s.. \End{j}]$ of $T[1.. \End{j}]$ whose RLE pseudo period is $\pi$
  by computing the longest common \emph{suffix} of
  $T[1.. \End{j-\sqrt{m}}]$ and $T[1.. \iota]$
  where $\iota$ is the ending position of the prefix-occurrence of the border of $T[\Beg{j - \sqrt{m} + 1}.. \End{j}]$.
  Such suffix can be computed by using the lcp data structure of Lemma~\ref{lem:rlefuncA} for the reversal of $T$.
  If $s = 1$, then we have $S_j =\$T[1.. \End{j}]$.
  Otherwise, we show the following claim:
  \begin{claim}\label{claim}
    The RLE pseudo period $\psi$ of $T[s-1.. \End{j}]$ is greater than $\sqrt{m}/2 -1$.
  \end{claim}
  \begin{proof}[Proof of Claim~\ref{claim}]
    Now we consider the RLE pseudo period $\psi$ of $T[s-1.. \End{j}]$ which satisfies $\psi > \pi$.
    For the sake of contradiction, we assume that $\psi \le \sqrt{m}/2 -1$ holds.
    Then, it holds that
    $\pi + 1 + \psi + 1 < \sqrt{m} \le r(T[s.. \End{j}])$ since $\psi > \pi$.
    Let $X=T[s..\End{j}]$ and $Y=T[s-1..\End{j}]$, and let $\mathsf{p}$ and $\mathsf{q}$ be their smallest periods, respectively.
    First observe that $\mathsf{q}\le |X|$: otherwise, $\mathsf{q}=|Y|$, and hence $Y$ is unbordered, implying $\psi = r(Y) \ge r(X)$, a contradiction.
    Further, let $\pi'$ and $\psi'$ be the RLE sizes of the length-$\mathsf{p}$ prefix of $X$ and the length-$\mathsf{q}$ suffix of $Y$, respectively.
    The latter is also a suffix of $X$ since $\mathsf{q} \le |X|$.
    Then, $\pi' \le \pi + 1$ and $\psi' \le \psi+1$ hold by the observation above.
    Thus, we obtain
    $\pi' + \psi' \le \pi + 1 + \psi + 1 < r(T[s.. \End{j}])$.
    If $\mathsf{p}+\mathsf{q}\ge|X|$, then the length-$\mathsf{p}$ prefix and the length-$\mathsf{q}$ suffix of $X$ cover $X$.
    It follows that $r(X)\le\pi'+\psi'$, a contradiction.
    Therefore, $\mathsf{p} + \mathsf{q} < |X|$ holds.
    Since $\mathsf{q}$ is also a period of $X$, $\gcd(\mathsf{p}, \mathsf{q})$ is a period of $X$
    from the periodicity lemma~\cite{fine1965uniqueness},
    where $\gcd(\mathsf{p}, \mathsf{q})$ is the greatest common divisor of $\mathsf{p}$ and $\mathsf{q}$.
    Since $\mathsf{p}$ is the smallest period of $X$, $\mathsf{q}$ is a multiple of $\mathsf{p}$.
    Then, $Y$ also has period $\mathsf{p}$, and hence its smallest period is $\mathsf{p}$.
    By the definition of the RLE pseudo period, the RLE pseudo period of $Y$ is therefore $\pi$, contradicting the definition of $s$.
  \end{proof}
  Namely, $S_j = T[s-1.. \End{j}]$ holds if $s \ge 2$.
  Finally, we compute each occurrence of $S_j$ in $O(\mathit{occ} + \log m)$ time
  by using the index structure based on tRLESA~\cite{TamakoshiGIBT15},
  where $\mathit{occ}$ is the number of occurrences of $S_j$ in $T$.
  The total reporting time is $O(\mathit{occ}+\log m)=O(\sqrt{m}+\log m)=O(\sqrt m)$ by Lemma~\ref{lem:occofSj}.
\end{proof}

\subsubsection{Preprocessing}
We construct
the data structure of Lemma~\ref{lem:datastrforRM}.

\begin{figure}[t!]
  \centering
  \includegraphics[width=0.9\linewidth]{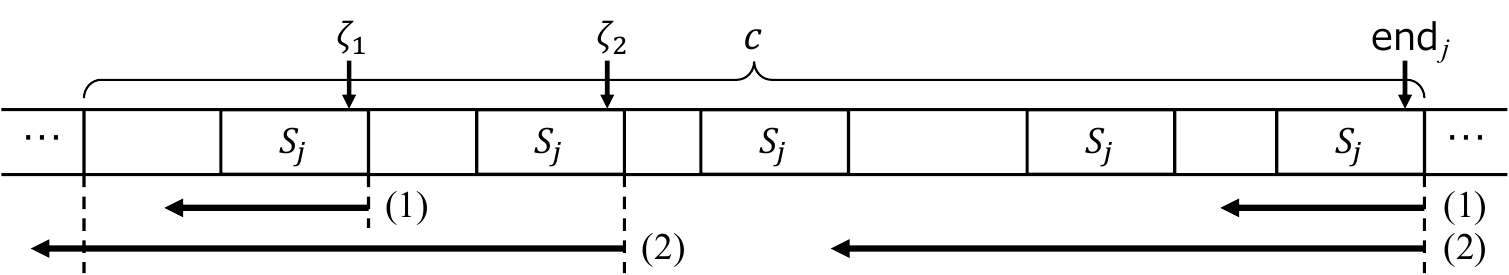} 
  \caption{We focus on a string $c \in C_j$ that is a candidate for an unbordered factor ending at $\End{j}$.
    (1) The lcs of $T[1.. \End{j}]$ and the prefix $T[1.. \zeta_1]$ of $T$ that ends at the leftmost occurrence of $S_j$ in the figure does not reach the starting position of $c$. We cannot yet determine whether $c$ is bordered.
    (2) The lcs of $T[1.. \End{j}]$ and the prefix $T[1.. \zeta_2]$ of $T$ that ends at the second leftmost occurrence of $S_j$ in this figure reaches the starting position of $c$, which reveals that $c$ is bordered.
  }\label{fig:LBcheck}
\end{figure}
\subsubsection{Query algorithm}
Again, let us consider the $k$th stage
and let $x = J_{k-1}.\first$, $y = J_k.\first$, and $z = J_k.\last$.
$D_k = T[\Beg{x}.. \End{z}]$ and $J_k = T[\Beg{y}.. \End{z}]$.
For each $j$ with $y \le j \le z$, 
we compute $S_j$ and its occurrences, and sort them
by using Lemma~\ref{lem:datastrforRM}.
Let $L_j$ be the sorted list of occurrences of $S_j$
whose ending positions are strictly smaller than $\End{j}$,
sorted in increasing order of their ending positions.
Given a set $C$ of candidates for the longest unbordered factors,
we first radix sort the elements of $C$ using
their starting points as the secondary key and
their ending points as the primary key.
Let $C'$ be the sorted list of the elements of $C$.
In the following, we scan elements in $C'$ in order throughout the algorithm.
We denote by $C_j$ the list of elements in $C'$ whose end positions are $\End{j}$.
Then, $C_j$ is a consecutive sub-list of $C'$ from the condition of the radix sorting.
For every $j$ with $y \le j \le z$,
we execute the following, scanning $C_j$ and $L_j$ from left to right:
Let variables $c$ and $\mathit{occ}$ store elements in $C_j$ and $L_j$ we focus on, respectively.
By using the data structure of Lemma~\ref{lem:rlefuncA} for the reversal of $T$,
compute the longest common suffix (lcs) of
$T[1.. \End{j}]$ and $T[1.. \zeta]$
where $\zeta$ is the ending position of $\mathit{occ}$~(see Fig.~\ref{fig:LBcheck}).
\begin{itemize}
  \item If the lcs does not  reach the starting position of $c$,
    then update $\mathit{occ}$ to the next element in $L_j$.
    If such an element in $L_j$ does not exist,
    then the current $c$ is the longest string in $C_j$
    that has no long border by Lemma~\ref{lem:Sjissuffix}.
  \item If the lcs reaches the starting position of $c$,
    then $c$ is bordered and hence
    update $c$ to the next element in $C_j$ that starts after $\zeta$.
    If such an element in $C_j$ does not exist,
    there is  no unbordered factor in $C_j$.
    So we update $c$ to the next element in $C'$ and increment $j$ accordingly.
\end{itemize}

The radix sort of $C$ takes $O(|C|)$ time.
Sorting the occurrence lists and scanning them with $C$ takes
$O(|C|+\sum_{j=y}^{z}|L_j|\log |L_j|)$.
By Lemma~\ref{lem:occofSj}, $|L_j|=O(\sqrt{m})$; there are $O(\sqrt{m})$ relevant values of $j$.
Hence $\sum_j|L_j|=O(m)$, and the total time per stage is $O(m\log m)$ with $O(m)$ working space.

Putting all pieces together, we complete the proof of Theorem~\ref{thm:main}, restated below:
\main*
\section{Average \texorpdfstring{$O(m\log^2 m)$}{O(mlog2m)}-time algorithm}\label{sec:mlog2m}
In this section, we prove Theorem~\ref{thm:mlog2m}:
\begin{restatable}{theorem}{mlogsqm}\label{thm:mlog2m}
  Fix a finite alphabet $\Sigma$ with $|\Sigma|\ge2$ and integers $m>1$ and $E\ge2$.
  Algorithm~\ref{alg:mlog2m} computes the set of all longest unbordered factors of every input string of RLE size $m$ using $O(m\log m)$ worst-case space and $O(m^2)$ worst-case time.
  When $T$ is drawn uniformly from the exponent-bounded RLE model $\mathcal T_{m,E}$ of Definition~\ref{def:boudedmodel}, its average-case running time is $O(m\log^2m)$.
\end{restatable}

Our algorithm in this section is based on the approach of
the average $O(n\log n)$-time algorithm proposed in~\cite{GawrychowskiKSS15}.
Here, we redefine the terms \emph{short border} and \emph{long border}.
In this section a border $b$ is short if $r(b)\le 6\lg m$, and long otherwise\footnote{$\lg$ denotes $\log_2$.}.
Note that these definitions of short/long borders are different from those used in Section~\ref{sec:algo}.

\begin{algorithm}[tb]
  \caption{Algorithm for computing the longest unbordered factor}
  \label{alg:mlog2m}
  \begin{algorithmic}[1]
    \Require String $T$ of RLE size $m$.
    \Ensure The set of longest unbordered factors of $T$.
    \If{\CheckLBf($T$) = TRUE}
    \State \Return $\LongUBf(T)$
    \EndIf
    \State $\LUB \gets \emptyset$ \Comment{$\LUB$: set of longest unbordered factors.}
    \State $\lub \gets 0$ \Comment{$\lub$: length of the longest unbordered factor.}
    \State Preprocess for {\CompLUBf}
    \For{$j \gets 2$ to $m$}
    \State $u \gets \CompLUBf(j)$
    \State $\ell \gets |u|$
    \If{$\ell < \lub$}
    \State \textbf{continue} \Comment{Do nothing and continue to the next stage.}
    \ElsIf{$\ell > \lub$}
    \State $\LUB \gets \emptyset$ \Comment{Clear the current tentative solutions.}
    \State $\lub \gets \ell$
    \EndIf
    \State $\LUB \gets \LUB \cup \{u\}$
    \EndFor
    \State  \Return $\LUB$
  \end{algorithmic}
\end{algorithm}

A high-level pseudocode of our algorithm is presented in Algorithm~\ref{alg:mlog2m}.
First, we call the subroutine $\CheckLBf(T)$,
which determines whether $T$ contains a factor with a long border.
If $T$ contains such a factor, we then call the subroutine $\LongUBf(T)$,
which na\"{i}vely computes the longest unbordered factors of $T$,
and output the result.
Otherwise, if $T$ contains no factors with a long border,
we initialize two variables $\LUB$ and $\lub$,
where $\LUB$ stores the current tentative solution, and $\lub$ stores the length of a string in $\LUB$.
The main part of Algorithm~\ref{alg:mlog2m} consists of a \textbf{for} loop for $j \in [2, m]$.
In each iteration,
we first call subroutine $\CompLUBf(j)$,
which computes the longest unbordered factor ending at position $\End{j}$.
We then update $\LUB$ and $\lub$ as necessary.
After the \textbf{for} loop terminates, the algorithm outputs the final answer $\LUB$.

Note that $\LongUBf(T)$ runs in $O(m^2)$ time using $O(m)$ space
by Corollary~\ref{cor:msquaredtime}.
We later show that $\CheckLBf(T)$ runs in $O(m\log m)$ time,
each $\CompLUBf(j)$ query runs in $O(\log^2 m)$ time,
and the preprocessing for $\CompLUBf$ takes $O(m\log^2 m)$ time~(see Sections~\ref{subsec:CheckLBf} and~\ref{subsec:CompLUBf}).
The total space complexity is $O(m\log m)$ due to the data structures required for $\CompLUBf$.
The discussions above imply the following lemma:
\begin{lemma}\label{lem:deterministic}
  Given an RLE-encoded string $\rle(T)$ of size $m$,
  $\textsc{check-long-border}(T)$ can be executed in $O(m\log m)$ time and $O(m)$ space.
  If it returns \textsc{false}, Algorithm~\ref{alg:mlog2m} runs in $O(m\log^2 m)$ time using $O(m\log m)$ space.
  If it returns \textsc{true}, Algorithm~\ref{alg:mlog2m} runs in $O(m^2)$ time using $O(m)$ space.
\end{lemma}

Finally, in Section~\ref{sec:probability}, we analyze the average case complexities of our method
under the \emph{exponent-bounded RLE model}, which will be defined in Definition~\ref{def:boudedmodel}.
The analysis concludes the proof of Theorem~\ref{thm:mlog2m}.

In the next two subsections, we describe how to implement subroutines $\CheckLBf$ and $\CompLUBf$.

\subsection{Implementation of CHECK-LONG-BORDER}\label{subsec:CheckLBf}
We first construct $\ExpSum$, $\tRLESA$, and $\tRLELCP$ in $O(m\log m)$ time using $O(m)$ space.
We then compute an array $\tRLELCPr$ where, for each $2 \le r \le m$,
$\tRLELCPr[r]$ stores the \emph{RLE size} of the lcp of $T\langle r\rangle$ and $T\langle r-1\rangle$.
Recall that
$T\langle r \rangle$ denotes the suffix $T[\End{\tRLESA[r]}.. n]$ and 
$\tRLELCP[r]$ stores the \emph{length} of the lcp of $T\langle r\rangle$ and $T\langle r-1\rangle$.
Given $\tRLELCP[r]$, we can compute $\tRLELCPr[r]$ in $O(\log m)$ time by performing binary search on $\ExpSum$.
Thus, array $\tRLELCPr$ can also be computed in $O(m\log m)$ time.
Here, the following lemma holds.
\begin{lemma}\label{lem:tRLELCPr}
  $\mathsf{tRLELCPr}$ stores a value greater than $6\lg m$ if and only if
  $T$ contains a factor with a long border.
\end{lemma}
\begin{proof}
  $(\Rightarrow)$
  Let $r$ and $k$ be integers satisfying $\tRLELCPr[r] = k$ and $k > 6\lg m$.
  Then, $T[\End{i}..\Beg{j+k-1}]$ has a long border, where
  $i=\min\{\tRLESA[r], \tRLESA[r-1]\}$ and $j=\max\{\tRLESA[r], \tRLESA[r-1]\}$.

  $(\Leftarrow)$ Suppose that $T$ contains a factor $L$ with a long border.
  By Observation~\ref{obs:contract},
  we can assume that $L$ starts at the ending position of some run.
  Let $i$ and $j$ be indices of runs such that
  $\End{i}$ is the starting position of $L$, and
  $\End{j}$ is the starting position of the suffix-occurrence of the longest border of $L$.
  Additionally, let $r_i = \tRLEISA[i]$ and $r_j = \tRLEISA[j]$.
  Then, the RLE size of the lcp of $T\langle r_i\rangle$ and $T\langle r_j\rangle$ is greater than $6\lg m$.
  Therefore, for all $r$ satisfying $\min\{r_i, r_j\} < r \le \max\{r_i, r_j\}$, it holds that $\mathsf{tRLELCPr}[r] > 6\lg m$.
\end{proof}

By Lemma~\ref{lem:tRLELCPr}, scanning $\mathsf{tRLELCPr}$ once,
we can determine whether $T$ contains a factor with a long border in $O(m)$ time.
The total time and space complexities of $\CheckLBf$ are $O(m\log m)$ time and $O(m)$ space, respectively.

\subsection{Implementation of COMPUTE-LUB}\label{subsec:CompLUBf}
This subsection describes an efficient implementation of $\CompLUBf$,
which represents the primary technical challenge of this section.
The objective of $\CompLUBf(j)$ is to obtain the longest unbordered factor of $T$ ending at position $\End{j}$,
under the assumption that \emph{$T$ contains no factors with a long border}.

\subsubsection{Preprocessing}
In the preprocessing~(at Line~6 of Algorithm~\ref{alg:mlog2m}), we construct
four arrays $\tRLESA$, $\tRLEISA$, $\tRLELCP$, and $\ExpSum$ of string $T$
and an RmQ data structure on $\tRLELCP$.
This preprocessing can be done in $O(m\log m)$ time using $O(m)$ space.

\subsubsection{Query algorithm}
In the rest of this section, we fix the index $j$ arbitrarily.
Recall that the longest unbordered factors must be RLE-bounded.
Thus, it suffices to check whether $T[\Beg{i}..\End{j}]$ is bordered
for all $1 \le i \le j$.
In our algorithm,
to determine whether $T[\Beg{i}..\End{j}]$ is bordered,
we first assess whether $T[\End{i}..\End{j}]$ is bordered.
On the one hand, by Lemma~\ref{lem:unbordered}, 
if $T[\End{i}..\End{j}]$ is unbordered, then $T[\Beg{i}..\End{j}]$ is also unbordered.
On the other hand, the converse does not hold in general, as we discuss later.
Because we assume that $T$ contains no factors with a long border,
the number of candidates for the (short) borders of each factor $T[\End{i}..\End{j}]$ is at most $6\lg m$.
Specifically, these candidates are the suffixes of $T[1.. \End{j}]$ starting at
the ending positions of the preceding $6 \lg m$ runs:
$T[\End{j-6\lg m+1}..\End{j}]$, $T[\End{j-6\lg m+2}..\End{j}]$, $\ldots$, $T[\End{j}.. \End{j}]$.

To retrieve occurrences of such candidates in $T$, we use the array $\tRLESA$ of $T$.
For $h \le j$, we denote by $\interval(h)$ the tRLESA interval for the factor $T[\End{h}..\End{j}]$.
Let $H_j=[\max\{j-6\lg m+1,1\}..j]$ and $\mathcal I=\{\interval(h)\mid h\in H_j\}$.
For each integer $h \in H_j$, we first obtain $r_h=\tRLEISA[h]$.
Next, we compute the length $\ell$ of $T[\End{h}.. \End{j}]$, as $\ell=\ExpSum[j]-\ExpSum[h]+1$.
Then, we compute the tRLESA interval $\interval(h) = [r_1, r_2]$ for $h$, which satisfies
$\lcp(T\langle r_1\rangle, T\langle r\rangle)\geq\ell$ and $\lcp(T\langle r\rangle, T\langle r_2\rangle)\geq\ell$.
This computation can be done in $O(\log m)$ time by combining binary search and RmQ on $\tRLELCP$.
Therefore, we can compute all the tRLESA intervals $\mathcal{I}$ in $O(\log^2 m)$ time.

For the computed tRLESA intervals $\mathcal{I}$, we show the next lemma:
\begin{lemma}\label{lem:nest}
  Any two tRLESA intervals $\interval(h_1), \interval(h_2) \in \mathcal{I}$ with $h_1 \ne h_2$ are non-crossing.
  Namely, it must hold that $\interval(h_1) \cap \interval(h_2) = \emptyset$ or $\interval(h_1) \subset \interval(h_2)$ or $\interval(h_2) \subset \interval(h_1)$.
\end{lemma}
\begin{proof}
  Let $B_1$ and $B_2$ be candidates for borders corresponding to $\interval(h_1)$ and $\interval(h_2)$, respectively.
  Without loss of generality, we assume $|B_1|>|B_2|$.
  For the sake of contradiction, we assume that $\interval(h_1) \cap \interval(h_2) \ne \emptyset$, $\interval(h_1) \not\subset \interval(h_2)$, and $\interval(h_2) \not\subset \interval(h_1)$.
  For any $r \in \interval(h_1) \cap \interval(h_2)$, both $B_1$ and $B_2$ are prefixes of $T\langle r\rangle$.
  Since $|B_1| > |B_2|$, $B_2$ is a prefix of $B_1$.
  On the other hand, for any $r' \in \interval(h_1)\setminus \interval(h_2)$, $B_1$ is a prefix of $T\langle r'\rangle$, while $B_2$ is not.
  This contradicts that $B_2$ is a prefix of $B_1$.
\end{proof}

For convenience, we define the set $\mathcal{C}$ of intervals
that is the smallest set satisfying
$\bigcup_{I\in\mathcal{I}\cup\mathcal{C}}I = [1, m]$.
Intuitively, $\mathcal{C}$ represents a complement of $\mathcal{I}$.
Since $|\mathcal{I}| \in O(\log m)$, the size of $\mathcal{C}$ is also $O(\log m)$.

Below, let us consider how to determine whether $T[\Beg{i}..\End{j}]$ is bordered for $i \le j$.
Let $r_i = \tRLEISA[i]$.
By the definition of $\mathcal{I}$,
$T[\Beg{i}..\End{j}]$ is unbordered if $r_i$ does not belong to any tRLESA intervals in $\mathcal{I}$
(or equivalently, $r_i \in I'$ for some $I' \in \mathcal{C}$).
Otherwise, $T[\Beg{i}..\End{j}]$ may be bordered.

For simplicity, we first consider a case
where there is exactly one tRLESA interval $\interval(h) \in \mathcal{I}$
such that $r_i \in\interval(h)$ and $h \ne j$.
Let $B = T[\End{h}.. \End{j}]$. Since $r_i\in \interval(h)$, string $B$ also occurs at position $\End{\tRLESA[r_i]} = \End{i}$, and thus, $B$ is a border of $T[\End{i}..\End{j}]$.
Now $T[\Beg{h}.. \End{j}] = B[1]^{\Exp{h}-1}B$ holds.
Thus, whether $T[\Beg{h}.. \End{j}]$ is a border of $T[\Beg{i}..\End{j}]$ depends on the exponents $\Exp{i}$ and $\Exp{h}$.
More precisely, in this special case,
$T[\Beg{i}..\End{j}]$ is unbordered if and only if $\Exp{i} > \Exp{h}$ holds~(see Fig.~\ref{fig:cond}).
\begin{figure}[tb]
  \centering
  \includegraphics[width=\linewidth]{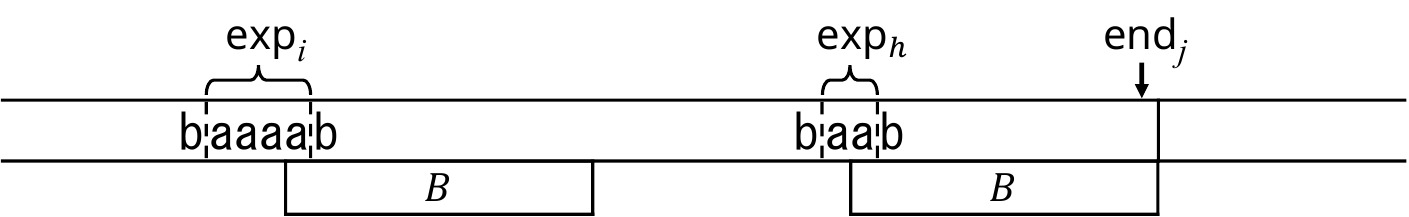} 
  \caption{
    While string $B = T[\End{h}.. \End{j}]$ is a border of $T[\End{i}..\End{j}]$, $B$'s left extension $T[\Beg{h}.. \End{j}]$ is not a border of $T[\Beg{i}..\End{j}]$.
    This is because $\Exp{i} > \Exp{h}$.
  }\label{fig:cond}
\end{figure}
As a corner case, consider the situation where $r_i \in \interval(j)$.
If $i=j$ and $\Exp{i}=1$, then $T[\Beg{i}..\End{j}]$ consists of a single character and is therefore unbordered;
we handle this case separately.
Otherwise, $r_i\in\interval(j)$ implies $\chara{i}=\chara{j}$, and thus $T[\Beg{i}..\End{j}]$ has a border,
regardless of the exponents $\Exp{i}$ and $\Exp{h}$.

The following notation is useful for handling multiple cases uniformly:
For an interval $I$, we denote
\begin{equation*}
  I.\Exp{} =
  \begin{cases}
    \infty  &\text{if } I = \interval(j) \in \mathcal{I},\\
    \Exp{h} &\text{if } I = \interval(h) \in \mathcal{I} \text{ for some }h\ne j,\\
    0       &\text{otherwise.}
  \end{cases}
\end{equation*}
Using this notation, we can summarize the above discussions in the following lemma
\begin{lemma}\label{lem:specialcase}
  Except for the separately handled case $i=j$ and $\Exp{j}=1$,
  suppose there is exactly one interval $I \in \mathcal{I} \cup \mathcal{C}$ such that $r_i \in I$.
  String $T[\Beg{i}..\End{j}]$ is unbordered if and only if $\Exp{i} > I.\Exp{}$.
\end{lemma}

Next, we consider the general case
where there are multiple intervals $I_1, I_2, \ldots, I_t$ such that
$r_i \in I_k \in \mathcal{I}$ for each $1\le k \le t$.
Without loss of generality, we can assume $I_1 \subset I_2 \subset \cdots \subset I_t$ by Lemma~\ref{lem:nest}.
As in the previous paragraph,
the condition $\Exp{i} > I_1.\Exp{}$ is required to ensure that $T[\Beg{i}..\End{j}]$ is unbordered.
Similarly, since 
$r \in I_1 \subset I_2 \subset \cdots \subset I_t$,
the conditions $\Exp{i} > I_k.\Exp{}$ are also required for \emph{all} $k$ with $1\le k \le t$.
Therefore, in this general case,
$T[\Beg{i}..\End{j}]$ is unbordered if and only if $\Exp{i} > \max\{I_k.\Exp{} \mid 1 \le k \le t \}$ holds~(see Fig.~\ref{fig:condnest} for the case $t = 2$).
This result generalizes the property established in Lemma~\ref{lem:specialcase}.
\begin{figure}[tb]
  \centering
  \includegraphics[width=\linewidth]{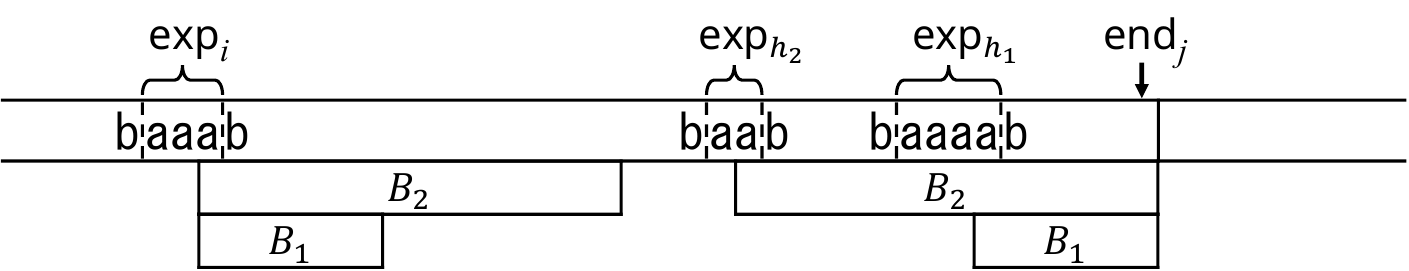} 
  \caption{
    Since $\Exp{i} > \Exp{h_2}$, $\mathtt{aa}B_2$ is not a border of $T[\Beg{i}..\End{j}]$.
    However, since $\Exp{i} \le \Exp{h_1}$, $\mathtt{aa}B_1$ is a border of $T[\Beg{i}..\End{j}]$,
  }\label{fig:condnest}
\end{figure}

From an algorithmic perspective,
we aim to avoid using the $\max$ function;
since $t \in O(\log m)$,
a na\"{i}ve calculation of the $\max$ function would require $O(\log m)$ time in the worst case.
To avoid this, we eliminate nested structures in $\mathcal{I}$.
We show the following lemma.
\begin{lemma}\label{lem:resolveNest}
  Given a set $\mathcal{I}_1 = \{(I_1, e_1), (I_2, e_2), \ldots, (I_N, e_N)\}$ of
  $N$ pairs of interval and integer
  such that any two distinct intervals are not crossing,
  we can transform $\mathcal{I}_1$ into a set $\mathcal{I}_2$ of size $O(N)$
  that satisfies the following three conditions (C1), (C2), and (C3)
  in $O(N^2)$ time using $O(N)$ space:
  \begin{description}
    \item[\bf (C1)] Intervals from $\mathcal{I}_2$ are mutually disjoint.
    \item[\bf (C2)] $\bigcup_{(I',e')\in\mathcal{I}_2}I'=\bigcup_{(I,e)\in\mathcal{I}_1}I$.
    \item[\bf (C3)] For any $(I',e')\in \mathcal{I}_2$ and $r\in I'$, $e'=\max\{e\mid(I,e)\in \mathcal{I}_1,r\in I\}$ holds.
  \end{description}
\end{lemma}
\begin{proof}
  First, we sort the elements in $\mathcal{I}_1$
  by increasing left endpoint, breaking ties by decreasing right endpoint.
  In what follows, we assume that $I_1, \ldots, I_N$ are already sorted.

  Next, we initialize $\mathcal{I}_2 = \emptyset$.
  For each element $(I_k, e_k)$ in $\mathcal{I}_1$, 
  we attempt to insert it into $\mathcal{I}_2$ while satisfying Condition (C1) as follows:
  Let $I_k = [r_1, r_2]$.
  \begin{enumerate}
    \item Search for a pair $([r_1', r_2'], e') \in \mathcal{I}_2$ with $[r_1, r_2]\subseteq[r_1', r_2']$.
    \item If there is no such pair, just insert $([r_1, r_2], e_k)$ into $\mathcal{I}_2$.
    \item Otherwise,
      delete $([r'_1,r'_2],e')$ from $\mathcal{I}_2$ and insert $([r_1,r_2],\max\{e_k,e'\})$.
      Also insert each of $([r'_1,r_1-1],e')$ and $([r_2+1,r'_2],e')$ if it is non-empty.
  \end{enumerate}
  Step 1 checks whether inserting $([r_1, r_2], e_k)$ leads to a nested structure.
  If it happens, $[r_1, r_2]$ is a sub-interval of the other one because
  we perform insertions in ascending order of the starting positions of the intervals.
  Furthermore, such a super-interval $[r_1', r_2']$ is uniquely determined.
  If $([r_1', r_2'], e')$ is not found, inserting  $[r_1, r_2]$ does not violate Condition (C1).
  Otherwise, we should eliminate nested structure appropriately.
  For Condition (C3), we insert $([r_1, r_2], \max\{e_k, e'\})$.
  To simulate this, we delete $([r'_1,r'_2],e')$, insert $([r_1,r_2],\max\{e_k,e'\})$,
  and insert the non-empty intervals among $([r'_1,r_1-1],e')$ and $([r_2+1,r'_2],e')$.

  Finally, $\mathcal{I}_2$ satisfies all the three conditions.
  Furthermore, $|\mathcal{I}_2| = O(N)$ because, in each iteration, the size of $\mathcal{I}_2$ can increase by at most two.
  It is easy to see that the algorithm runs in $O(N^2)$ time using $O(N)$ space\footnote{With a detailed amortized analysis,
  the time complexity can be improved to $O(\mathsf{sort}(N) + N)$. However, $O(N^2)$ time is fast enough for our specific purpose.}.
\end{proof}

According to Lemma~\ref{lem:resolveNest},
we can transform $\mathcal{I}$ into $\mathcal{I}'$ in $O(\log^2 m)$ time
such that 
$\mathcal{I}'$ satisfies the three conditions stated in the lemma.
Similar to Lemma~\ref{lem:specialcase}, the following lemma holds.
\begin{lemma}\label{lem:keylemma}
  Except for the separately handled case $i=j$ and $\Exp{j}=1$,
  string $T[\Beg{i}..\End{j}]$ is unbordered if and only if
  $\Exp{i} > I.\Exp{}$ holds, where $I$ is the unique interval in $\mathcal{I'}\cup\mathcal{C}$
  that contains $\tRLEISA[i]$.
\end{lemma}

Recall that $\CompLUBf(j)$ returns
the longest unbordered suffix of $T[1.. \End{j}]$.
Specifically, it suffices to compute 
$\min\{i \le j \mid T[\Beg{i}..\End{j}] \text{ is unbordered.}\}$
if such an index exists.
By Lemma~\ref{lem:keylemma}, this minimum can be written as 
\[
  \min_{I\in\mathcal{I}'\cup\mathcal{C}}\{ i \mid \tRLEISA[i] \in I \text{ and } \Exp{i} > I.\Exp{}. \}.
\]
If $\Exp{j}=1$, we additionally include $j$ as a candidate for this minimum,
since $T[\Beg{j}..\End{j}]$ consists of a single character and is unbordered.
This minimization problem can be reduced to 2D range minimum queries (2D-RmQs) as follows.
Given the array $\tRLESA$ of $T$,
for every $1 \le i \le m$,
we make a weighted point $(\tRLEISA[i], \Exp{i})$ with weight $i$ on $m\times n$ grid.
For every interval $I = [r_1, r_2] \in \mathcal{I'}\cup\mathcal{C}$,
find the smallest weight within the 2D orthogonal range $[r_1, r_2]\times[I.\Exp{}+1, \infty]$~(see Fig.~\ref{fig:orthogonalRmQ}).
The global minimum weight yields the desired value.
\begin{figure}[tb]
  \centering
  \includegraphics[width=0.7\linewidth]{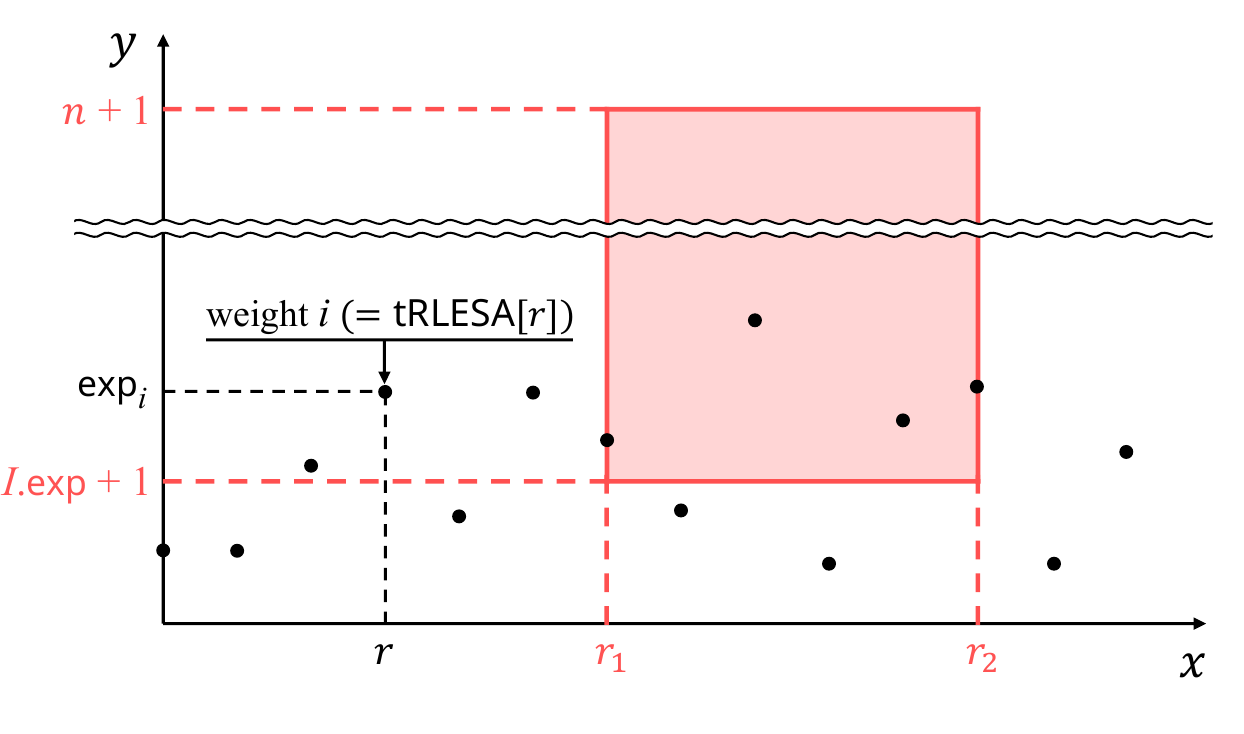}
  \caption{
    Illustration for a reduction to 2D-RmQ.
    For each $1\le i \le m$, 
    a point $(x, y)=(r, \Exp{i})$ with weight $i$ is defined such that $\tRLESA[r] = i$.
    The rectangle represents the query range $[r_1, r_2]\times[I.\Exp{}+1, n+1]$,
    where $n+1$ is used in place of $\infty$.
  }\label{fig:orthogonalRmQ}
\end{figure}
A 2D-RmQ for $m$ points can be answered in $O(\log m)$ time after $O(m\log^2 m)$ time preprocessing
using 2D range trees~\cite{BergCKO08}.
Note that the preprocessing for the 2D-RmQ is executed at Line~6 of Algorithm~\ref{alg:mlog2m}.

We summarize the query algorithm for $\CompLUBf(j)$ as follows:
First, for all candidates for the short borders of $T[1.. \End{j}]$,
we compute their corresponding tRLESA intervals.
We then obtain the sets $\mathcal{I}$ and $\mathcal{C}$ in $O(\log^2 m)$ time.
Next, $\mathcal{I}$ is transformed into $\mathcal{I}'$ in $O(\log^2 m)$ time
to remove nested structures.
Finally, using the 2D-RmQ data structure, we find the minimum index $i$ such that
$T[\Beg{i}.. \End{j}]$ is unbordered in $O(\log^2 m)$ time.

Consequently, the total query time complexity is $O(\log^2 m)$.
The bottleneck for space complexity is the 2D-RmQ data structure, which requires $O(m\log m)$ space.

Wrapping up the discussions above, we obtain Lemma~\ref{lem:deterministic}. 

\subsection{Average case analysis}\label{sec:probability}

This subsection shows the probability that $\CheckLBf(T) =$ \textrm{TRUE} is less than $\frac{2}{m}$ in Lemma~\ref{lem:prob}.
Consequently, the \emph{average} running time of Algorithm~\ref{alg:mlog2m} is $O(m\log^2 m)$ by Lemma~\ref{lem:deterministic}.

To discuss the average time complexity,
we define a probability space and a distribution tailored for RLE strings.
\begin{definition}[Exponent-bounded RLE model]\label{def:boudedmodel}
  Fix a finite alphabet $\Sigma$ with $|\Sigma| \ge 2$.
  For positive integers $m$ and $E$, let
  $\mathcal{T}_{m,E} = \{ c_1^{e_1} c_2^{e_2}\cdots c_m^{e_m} \mid c_i\ne c_{i+1}\text{ and } 1\le e_i \le E \}$.
  In the \emph{exponent-bounded RLE model}, the input string $T$ is chosen uniformly at random from $\mathcal{T}_{m,E}$.
\end{definition}
In other words, the exponent values $e_1,\ldots,e_m$ are chosen independently and uniformly from $[1, E]$,
independently of the run-characters.
Thus, it holds that $|\mathcal{T}_{m,E}| = |\Sigma|(|\Sigma|-1)^{m-1}E^m$.
Under the bounded-exponent RLE model, we show the following.
\begin{lemma}\label{lem:prob}
  Let $\Sigma$ be a non-unary alphabet.
  Further let $m>1$ and $E\geq 2$.
  If $T$ is chosen according to the exponent-bounded RLE model,
  then the probability that $T$ contains a factor with a long border is less than $2/m$.
\end{lemma}

Before proving Lemma~\ref{lem:prob}, we first show the following lemma:
\begin{lemma}\label{lem:3lgm}
  Let $\Sigma$ be a non-unary alphabet,
  and let $m > 1$ be the RLE size of a string $T \in \Sigma^\star$.
  If $T$ contains a factor with a long border, then there is a factor $w$ of $T$ such that
  $r(w) = 3\lg m$ and $w$ has two non-overlapping occurrences in $T$.
  In other words, there exist strings $x, y, z, w \in \Sigma^\star$ such that
  $r(w) = 3\lg m$ and $T = xwywz$.
\end{lemma}
\begin{proof}
  Suppose that $T$ contains a factor $S$ having a long border $B$.
  Let $B_{\mathrm{pre}}$ and $B_{\mathrm{suf}}$ be the prefix- and suffix-occurrences of $B$ in $S$, respectively.
  Let $k=3\lg m$.
  Since $B$ is a long border, we have $r(B)>6\lg m=2k$.

  Let $w_{\mathrm{pre}}$ be the prefix of $B$ consisting of its first $k$ runs,
  and let $w_{\mathrm{suf}}$ be the suffix of $B$ consisting of its last $k$ runs.
  Since $r(B)>2k$, these two factors are disjoint within $B$.
  Therefore, $|w_{\mathrm{pre}}|+|w_{\mathrm{suf}}|\le |B|$.
  Let $w$ be the shorter of $w_{\mathrm{pre}}$ and $w_{\mathrm{suf}}$, breaking ties arbitrarily.
  Then, $r(w)=k$ and $|w|\le \frac{|B|}{2}$.

  Let $p=|S|-|B|$.
  Since $B$ is a border of $S$, $p$ is a period of $S$.
  We define $d=\lceil (|w|/p) \rceil p$.
  Then, $d$ is a multiple of $p$ and satisfies $|w|\le d\le |w|+p$.

  If $w=w_{\mathrm{pre}}$, consider its occurrence at the beginning of $B_{\mathrm{pre}}$.
  By the $p$-periodicity of $S$, shifting this occurrence by $d$ positions to the right gives another occurrence of $w$.
  If $w=w_{\mathrm{suf}}$, consider its occurrence at the end of $B_{\mathrm{suf}}$.
  Symmetrically, shifting this occurrence by $d$ positions to the left gives another occurrence of $w$.
  In either case, the two occurrences do not overlap because $d\ge |w|$.
  Moreover, $d+|w|\le 2|w|+p\le |B|+p=|S|$,
  and hence both occurrences lie entirely within $S$.

  Thus, $w$ has two non-overlapping occurrences in the factor $S$ of $T$.
  Consequently, there exist strings $x,y,z\in\Sigma^\star$ such that $T=xwywz$.
\end{proof}

We are ready to prove Lemma~\ref{lem:prob}.

\begin{proof}[Proof of Lemma~\ref{lem:prob}]
  By Lemma~\ref{lem:3lgm}, it suffices to obtain an upper bound for the probability that $T$ can be decomposed into $xwywz$ with $r(w)=3\lg m$.

  Let $e$ be the length of the first run of $w$.
  By deleting the first $e-1$ characters from both occurrences, we obtain two equal non-overlapping factors of RLE size $3\lg m$ that start at ending positions of runs of $T$.
  After absorbing the deleted prefixes into $x$ and $y$ and renaming the shortened factor as $w$,
  we may assume that the two occurrences start at $\End{i}$ and $\End{j}$ for some run indices $i<j$.

  Let $\ell=3\lg m$.
  Fix two run indices $i<j$, and suppose that two equal non-overlapping
  factors of RLE size $\ell$ start at $\End{i}$ and $\End{j}$, respectively.
  The equality of these factors forces $\Exp{j+t}=\Exp{i+t}$ for every $1\le t\le \ell-2$,
  because these exponents correspond to the complete internal runs of the two occurrences.
  Moreover, $i+\ell-1\le j$ holds since the factors do not overlap.
  Thus, conditioned on the first $j$ runs, the values $\Exp{i+1},\ldots,\Exp{i+\ell-2}$ are already fixed, whereas $\Exp{j+1},\ldots,\Exp{j+\ell-2}$ are independent and uniformly distributed over $[1,E]$.
  Therefore, the probability for each fixed pair $(i,j)$ is at most
  $E^{-(\ell-2)} \le 2^{-(\ell-2)} \le \frac{4}{m^3}$.
  Since there are fewer than $\binom{m}{2}$ possible pairs, the union bound gives a probability of less than
  $\binom{m}{2}\frac{4}{m^3}<\frac{2}{m}$.
\end{proof}

Combining Lemma~\ref{lem:prob} with Lemma~\ref{lem:deterministic}, we obtain Theorem~\ref{thm:mlog2m}, restated below:

\mlogsqm*

\section{Conclusions and future work}

In this paper, we addressed the problem of computing all longest unbordered factors of a run-length encoded (RLE) string of size $m$.
Based on the observation that
any longest unbordered factor must start at the beginning position of a run and end at the ending position of a run,
we developed an $O(m \sqrt{m} \log^2 m)$-time algorithm using $O(m \log^2 m)$ space.
Furthermore,
we proposed an alternative algorithm that achieves an average time complexity of $O(m \log^2 m)$ using $O(m \log m)$ space,
providing a more efficient solution in terms of average time.
Although our strategies are inspired by the prior algorithms by Gawrychowski et al.~\cite{GawrychowskiKSS15},
the algorithmic details differ, requiring specific techniques tailored to the properties of unbordered factors in RLE strings.

A primary direction for future work is to develop a faster algorithm with 
a deterministic $O(m\, \mathrm{polylog}(m))$ \emph{worst-case} time complexity.
Another promising avenue is to extend this problem to strings compressed using other schemes,
such as grammar compression or run-length Burrows–Wheeler transform.

\subsection*{Acknowledgements}
The authors thank Hiro Ito and Suguru Hayata~(The University of Electro-Communications)
for their helpful discussions and constructive comments on this work.
Takuya Mieno was supported by JSPS KAKENHI Grant Number JP24K20734.

\end{document}